\documentclass[11pt]{article}

\usepackage{tikz}
\usetikzlibrary{matrix,arrows,shapes}
\usepackage{pgf}
\usepackage{url}
  \usepackage{enumerate}
  \usepackage{amsthm,amssymb,amsmath,amscd}
    \usepackage[retainorgcmds]{IEEEtrantools}
  \usepackage{eucal}
  \usepackage{mathrsfs}
  \usepackage{mathtools}
  \usepackage{upgreek}
   \usepackage{dsfont}
      \usepackage{yfonts}
  \usepackage[T1]{fontenc}
  \usepackage{bbm}
  \usepackage{geometry}
\usepackage{array}
\usepackage{booktabs}
\usepackage[retainorgcmds]{IEEEtrantools}
  \usepackage{ulem}
  \usepackage[colorinlistoftodos,textsize=footnotesize]{todonotes}
  
\definecolor{see}{RGB}{67,75,179}
\definecolor{darksee}{RGB}{42,44,148}
\definecolor{honey}{RGB}{232,180,129}
\definecolor{lighthoney}{RGB}{255,254,220}
\definecolor{citecol}{rgb}{0.5,0,0} 
\usepackage[a4paper]{hyperref}
\hypersetup
{
bookmarks=true,
bookmarksopen=true,
pdftitle="BV formalism",
pdfauthor="Katarzyna Rejzner", 
pdfsubject="BV formalism", 
pdfmenubar=true, 
pdfhighlight=/O, 
colorlinks=true, 
pdfpagemode=None, 
pdfpagelayout=SinglePage, 
pdffitwindow=true, 
linkcolor=see, 
citecolor=red!75!darkgray, 
urlcolor=see 
}
\makeatletter
\newcommand\bibname{References}%
\def\bibfont{\fontsize{10}{11}\selectfont}
\newdimen\bibindent
\def\@biblabel#1{[#1]}
\renewenvironment{thebibliography}[1]
     {\section*{\bibname}\bibfont%
      \list{\@biblabel{\arabic{enumiv}}}%
           {\settowidth\labelwidth{\@biblabel{#1}}%
            \leftmargin\labelwidth
	    \labelsep6pt
            \advance\leftmargin\labelsep
	    \itemsep3pt\parsep2pt
            \usecounter{enumiv}%
            }%
      \sloppy\clubpenalty4000\widowpenalty4000%
      \sfcode`\.=\@m}
     {\def\@noitemerr
       {\@latex@warning{Empty `thebibliography' environment}}%
      \endlist}
\makeatother
\newcommand{\Tens}{\mathfrak{Tens}}
\newcommand{\CE}{\mathfrak{CE}}  
\newcommand{\E}{\mathfrak{E}}
\newcommand{\C}{\mathfrak{C}}

\newcommand{\BV}{\mathfrak{BV}}
\newcommand{\fA}{\mathfrak{A}}

\newcommand{\D}{\mathfrak{D}}

\newcommand{\F}{\mathfrak{F}}

\newcommand{\fG}{\mathfrak{G}}

\newcommand{\Lor}{\mathfrak{Lor}}

\newcommand{\frakg}{\mathfrak{g}}
\newcommand{\X}{\mathfrak{X}}

\newcommand{\CEcal}{\mathcal {CE}}

\newcommand{\Gcal}{\mathcal{G}}  

\newcommand{\Kcal}{\mathcal{K}}  
\newcommand{\Ccal}{\mathcal{C}}
\newcommand{\Dcal}{\mathcal{D}}
\newcommand{\Ecal}{\mathcal{E}} 
\newcommand{\Fcal}{\mathcal{F}} 
\newcommand{\BVcal}{\mathcal{BV}} 

\newcommand{\Ncal}{\mathcal{N}}
\newcommand{\Mcal}{\mathcal{M}}
\newcommand{\Ocal}{\mathcal{O}}
\newcommand{\Scal}{\mathcal{S}}

\newcommand{\Tcal}{\mathcal{T}}
\newcommand{\Vcal}{\mathcal{V}}

\newcommand{\Zcal}{\mathcal{Z}}

\newcommand{\Ci}{\mathcal{C}^\infty} 

\newcommand{\obj}{\mathrm{Obj}}

\newcommand{\Nat}{\mathrm{Nat}}

\newcommand{\Loc}{\mathrm{\mathbf{Loc}}}       
\newcommand{\Obs}{\mathrm{\mathbf{Obs}}}       
\newcommand{\Top}{\mathrm{\mathbf{Top}}}          
\newcommand{\Vect}{\mathrm{\mathbf{Vec}}}       

\newcommand{\CnMfd}{\mathrm{\mathbf{CnMfd}}}
\newcommand{\LcMfd}{\mathrm{\mathbf{LcMfd}}}

\newcommand{\WF}{\mathrm{WF}}         
\newcommand{\id}{\mathrm{id}}               
\DeclareMathOperator{\supp}{\mathrm{supp}}      
\newcommand{\relsupp}{\mathrm{rel\, supp}}      
\newcommand{\Diff}{\mathrm{Diff}}        
\newcommand{\im}{\mathrm{Im}}             
\newcommand{\tr}{\mathrm{tr}}                 
\newcommand{\loc}{\mathrm{loc}}
\newcommand{\inv}{\mathrm{inv}}
\newcommand{\reg}{\mathrm{reg}}

\newcommand{\pg}{\mathrm{pg}}

\newcommand{\af}{\mathrm{af}}
\newcommand{\ta}{\mathrm{ta}}
\newcommand{\gh}{\mathrm{gh}}

\newcommand{\mc}{{\mu\mathrm{c}}}
\newcommand{\ml}{\mathrm{ml}}
\newcommand{\ex}{\mathrm{ext}}

\newcommand{\NN}{\mathbb{N}}          
\newcommand{\RR}{\mathbb{R}}           
\newcommand{\M}{\mathbb{M}} 	     
\newcommand{\al}{\alpha}
\newcommand{\bet}{\beta}

\newcommand{\Ga}{\Gamma}

\newcommand{\De}{\Delta}

\newcommand{\la}{\lambda}
\newcommand{\La}{\Lambda}

\newcommand{\ph}{\varphi}
\newcommand{\T}{\cdot_{{}^\Tcal}}

\newcommand{\TR}{\cdot_{{}^{\TTR}}}

\newcommand{\TT}{\Tcal}
\newcommand{\TTR}{\Tcal_H}

\newcommand{\qme}{{\textsc{qme}}}

\newcommand{\cme}{{\textsc{cme}}}
\newcommand{\mwi}{{\textsc{mwi}}}
\newcommand{\di}{\textrm{div}}
\newcommand{\ds}{\textrm{div}^*}

\newcommand{\sst}[1]{\scriptscriptstyle{#1}}  
\newcommand{\vr}[1]{\boldsymbol{#1}}         
\newcommand{\hinv}{*^{\!\sst{-\!1}}}             
\newcommand{\minus}{\sst{-1}}   
\newcommand{\be}{\begin{equation}}
\newcommand{\ee}{\end{equation}}

\newcommand{\Lap}{\bigtriangleup}

\newcommand{\os}{\stackrel{\mathrm{o.s.}}{=}}

\newcommand{\otoprule}{\midrule[\heavyrulewidth]}
\newcommand{\Pei}[2]{\lfloor #1, #2 \rfloor}
\newcommand{\skal}[2]{\left< #1 , #2 \right>}

\newcommand{\dgr}{{\sst\ddagger}}

\def\normOrd#1{\mathop{:}\nolimits\!#1\!\mathop{:}\nolimits}
\author{\null\\Romeo Brunetti$^{(1)}$, Klaus Fredenhagen$^{(2)}$, Katarzyna Rejzner$^{(3)}$ \\
  \null\\
  \null\\
  \small{$^{(1)}$Department of Mathematics, University of Trento}\\
        \small{$^{(2)}$ II Institute for Theoretical Physik, University of Hamburg}\\
\small{$^{(3)}$ Department of Mathematics, University of York}\\
\small{\texttt{romeo.brunetti@unitn.it,
klaus.fredenhagen@desy.de,}}\\ 
\small{\texttt{kasia.rejzner@york.ac.uk }}}

\title{Quantum gravity from the point of view of \\locally covariant quantum field theory }

\begin{document}
 \sloppy
\date{}
 \maketitle

  \theoremstyle{plain}
  \newtheorem{df}{Definition}[section]
  \newtheorem{thm}[df]{Theorem}
  \newtheorem{prop}[df]{Proposition}
  \newtheorem{cor}[df]{Corollary}
  \newtheorem{lemma}[df]{Lemma}
  \newtheorem{exa}[df]{Example}
  
  \theoremstyle{plain}
  \newtheorem*{Main}{Main Theorem}
  \newtheorem*{MainT}{Main Technical Theorem}

  \theoremstyle{definition}
  \newtheorem{rem}[df]{Remark}
 \theoremstyle{definition}
  \newtheorem{ass}{\underline{\textit{Assumption}}}[section]

\hspace{5mm} 
\begin{center}
\textit{Dedicated to Roberto Longo on the occasion of his 60\textsuperscript{th} birthday}\\
\end{center}
\hspace{1mm} 
\begin{abstract}
We construct perturbative quantum gravity in a generally covariant way. In particular our construction is background independent. It is based on the locally covariant approach to quantum field theory and the renormalized Batalin-Vilkovisky formalism. We do not touch the problem of nonrenormalizability and interpret the theory as an effective theory at large length scales. 
\end{abstract}
\tableofcontents
\markboth{Contents}{Contents}
\section{Introduction}
The incorporation of gravity into quantum theory is one of the great challenges of physics. The last decades were dominated by attempts to reach this goal by rather radical new concepts, the best known being string theory and loop quantum gravity. A more conservative approach via quantum field theory was originally considered to be hopeless because
of severe conceptual and technical problems. In the meantime it became clear that also the other attempts meet enormous problems, and it might be worthwhile to reconsider the quantum field theoretical approach. Actually, there are indications that the obstacles in this approach are less heavy than originally expected. 

One of these obstacles is perturbative non-renormalizability \cite{Sagnotti,vandeVen:1991gw} which actually means that the counter-terms arising in higher order of perturbation theory cannot be taken into account by readjusting the parameters in the Lagrangian. Nevertheless, theories with this property can be considered as effective theories with the property that only finitely many parameters have to be considered below a fixed energy scale \cite{Weinberg}. Moreover, it  may be that the theory is actually asymptotically safe in the sense that there is an ultraviolet fixed point of the renormalisation group flow with only finitely many relevant directions \cite{Weinberg79}. Results supporting this perspective have been obtained by Reuter et al. \cite{Reu98,Reu02}.

Another obstacle is the incorporation of the principle of general covariance. Quantum field theory is traditionally based on the symmetry group of Minkowski space, the Poincar\' {e} group. In particular, the concept of particles with the associated notions of a vacuum (absence of particles) and scattering states heavily relies on Poincar\' {e} symmetry.
Quantum field theory on curved spacetime which might be considered as an intermediate step towards quantum gravity already has no distinguished particle interpretation.
In fact, one of the most spectacular results of quantum field theory on curved spacetimes is Hawking's prediction of black hole evaporation \cite{Hawking:1974sw}, a result which may be understood  as a consequence of different particle interpretations in different regions of spacetime. (For a field theoretical derivation of the Hawking effect see \cite{FH}.)

Quantum  field theory on curved spacetime is nowadays well understood. This success is based on a consequent use of appropriate concepts. First of all, one has to base the theory on the principles of algebraic quantum field theory since there does not exist a distinguished Hilbert space of states. In particular, all structures are formulated in terms of local quantities. Global properties of spacetime do not enter the construction of the algebra of observables. They become relevant in the analysis of the space of states whose interpretation up to now is less well understood. It is at this point where the concept of particles becomes important if the spacetime under consideration has asymptotic regions similar to Minkowski space. Renormalization can be done without invoking any regularization by the methods of causal perturbation theory \cite{EG}. Originally these methods made use of properties of a Fock space representation, but could  be generalized to a formalism based on algebraic structures on a space of functionals of classical field configurations where the problem of singularities  can be treated by methods of microlocal analysis \cite{BFK95,BF0,HW}. The lack of isometries in the generic case could be a problem for a comparison of renormalisation conditions at different  points of spacetime. But this problem could be overcome by requiring local covariance, a principle, which relates theories at different spacetimes. The arising theory is already generally covariant and includes all typical quantum field theoretical models with the exception of supersymmetric theories (since supersymmetry implies the existence of a large group of isometries (Poincar\' {e} group or Anti de Sitter group)). See \cite{BFV,BDF} for more details.

It is the aim of this paper to extend this approach to gravity. But here there seems to be a conceptual obstacle. As discussed above, a successful treatment of quantum field theory on generic spacetimes requires the use of local observables, but unfortunately there are no diffeomorphism invariant localized functionals of the dynamical degrees of freedom (the metric in pure gravity). 
Actually, this creates in addition to technical complications also a problem for the interpretation. Namely, Nakanishi \cite{Naka,Nak78} uses the distinguished background for a formal definition of an S-matrix, and one could base an interpretation of the formalism in terms of the S-matrix provided it exists. But an interpretation based on the S-matrix is no longer possible for generic backgrounds.
 Often this difficulty is taken as an indication that a quantum field theoretical treatment of quantum gravity is impossible. We propose a solution of this problem by the concept of relative observables introduced by Rovelli in the framework of loop quantum gravity \cite{Rovelli:2001bz} and later used and further developed  in \cite{Dittrich:2005kc,Thiemann:2004wk}.
The way out is to replace the requirement of invariance by covariance. We associate observables to spacetime subregions in a locally covariant way (compare with \cite{BFV,HW}). Such observables transform equivariantly under diffeomorphism transformations, but the relations between them are diffeomorphism invariant.

Because of its huge group of symmetries the quantization of gravity is plagued by problems known from gauge theories, and a construction seems to require the introduction of redundant quantities which at the end have to be removed. In perturbation theory the Batalin-Vilkovisky (BV) approach \cite{Batalin:1977pb,Batalin:1981jr} has turned out to be the most systematic method, generalizing the BRST approach \cite{BRST1,BRST2,Tyu}. 
In a previous paper \cite{FR} two of us performed this construction for classical gravity, and in another paper \cite{FR3} we developed a general scheme for a renormalized BV formalism for quantum physics, based on previous work of Hollands on Yang-Mills theories on curved spacetimes \cite{H}  and of Brennecke and D\"utsch on a general treatment of anomalies \cite{BreDue}. In the present paper it therefore suffices to check whether the assumptions used in the general formalism are satisfied in gravity.

In  the BV approach one constructs at the end the algebra of observables as a cohomology of a certain differential. But here the absence of local observables shows up in the triviality of the corresponding cohomology, as long as one restricts the formalism to local functionals of the perturbation metric on a fixed spacetime. A nontrivial cohomology class arises on the level of locally covariant fields which are defined simultaneously on all spacetimes. This is solved by relaxing the locality assumption a bit, and considering the relational observables.  

The paper is organized as follows. We first describe the functional framework for classical field theory adapted to gravity. This framework was in detail developed in \cite{BFR} but many ideas may already be found in the work of DeWitt \cite{DeWitt:2003pm}, and an earlier version  is \cite{DF}. In this framework, many aspects of quantum gravity can be studied, in particular the gauge symmetry induced by general covariance. 

As already discussed in \cite{FR}, the candidates for local observables are locally covariant fields which act simultaneously on all spacetimes in a coherent way. Mathematically, they can be defined as natural transformations between suitable functors (see \cite{BFV}). It seems, however, difficult to use them directly as generators of an algebra of observables for quantum theory (for attempts see \cite{Few} and \cite{Rej11b,FR}). Moreover, the action of the BV operator on such locally covariant quantum fields $\Phi$ involves an additional term, which cannot be generated by the antibracket \cite{FR}. We therefore take a different path here and, on a generic background spacetime $\Mcal=(M,g_0)$, we evaluate fields $\Phi_{\Mcal}$ on test functions of the form $f=\vr{f}\circ X_{g_0+h}$, where in the simplest situation $\vr{f}:\RR^4\rightarrow \RR$ and $X_{g_0+h}^{\mu}$, $\mu=0,\ldots,3$ are coordinate fields constructed as scalar curvature invariants depending on the  full metric $g=g_0+h$. We interpret the obtained diffeomorphism invariant quantities as relative observables, similar to concepts developed in loop quantum gravity \cite{Rovelli:2001bz,Dittrich:2005kc,Thiemann:2004wk}. 

More generally, in the absence of an intrinsic choice of a coordinate system the physical interpretation is based on the relations between different observables. In suitable cases some of them could be thought of as coordinates but this is not necessary for a physical interpretation. This variant of the proposed formalism is discussed in section \ref{abstract:view}.

The algebra generated by the relative observables is subsequently quantized with the use of the BV formalism. For the purposes of perturbation theory we replace the diffeomorphism group by the Lie algebra of vector fields, so the ``gauge invariance'' is in our framework the invariance under infinitesimal diffeomorphisms realized through the Lie derivative. The quantization proceeds following the paradigm proposed in \cite{FR3}. Firstly, we extend the algebra of relative observables with auxiliary objects like ghosts, antifields, etc. and add appropriate terms to the action (section \ref{BVcomplex}). The final outcome of this procedure is a graded differential algebra $(\BVcal(\Mcal),s)$, where $s$ is the classical BV differential, and the extended action $S^\ex$ such that $s$ is locally generated by the antibracket (the Schouten bracket on $\BVcal
(\Mcal)$). In section \ref{quantization} we quantize the extended theory using methods of perturbative algebraic quantum field theory (pAQFT).  In the intermediate steps we need to split the interaction (around the background metric $g_0$) into the free part $S_0$ and the interaction term $S_I$. First, we quantize the free part by choosing a Hadamard solution of the linearized Einstein equation. We then can apply the renormalized BV formalism as developed in \cite{FR3}. A crucial role is played by the M{\o}ller map which maps interacting fields to free ones. In particular it also intertwines the free BV differential with that of the interacting theory.

We then show that the theory is background independent (section \ref{background}), in the sense that a localized change in the background which formally yields an automorphism on the algebra of observables (called relative Cauchy evolution in \cite{BFV}) is actually trivial, in agreement with the proposal made in \cite{BF1} (see also \cite{FR2}).

We sketch how to construct states on the algebra of observables, using the perturbative ansatz of \cite{DF99}. In the first step one constructs a pre-Hilbert representation of linearized theory and the subspace of vectors with positive inner product is distinguished as the cohomology of the free BRST charge $Q_0$. We refer to the literature where such construction was achieved on some special classes of spacetimes \cite{Fewster:2012bj,BDM}. In the next step we construct the representation of the full theory on the space $\Kcal$ of formal power series in $\hbar$ and the coupling constant $\la$ with coefficients on $\Kcal_0$. The positive subspace is then recovered as the cohomology of the full interacting BRST charge as proposed in \cite{DF99}. The consistency of this approach with the BV formalism has been discussed in \cite{Rej13}.
\section{Classical theory}    
\subsection{Configuration space of the classical theory}\label{conf:space}
We start with defining the kinematical structure which we will use to describe the gravitational field. We follow \cite{FR}, where the classical theory was formulated in the locally covariant framework. To follow this approach we need to define some categories. Let $\Loc$ be the category of  time-oriented globally hyperbolic spacetimes with causal isometric embeddings
 as morphisms. The configuration space of classical gravity is a subset of the space of Lorentzian metrics, which can be equipped with an infinite dimensional  manifold structure. To formulate this in the locally covariant framework we need to introduce a category, whose objects are infinite dimensional manifolds and whose arrows are smooth injective linear maps. There are various possibilities to define this category. One can follow \cite{Ham} and use the category  $\LcMfd$ of differentiable manifolds modeled on locally convex vector spaces or use the more general setting of convenient calculus, proposed in \cite{Michor}. The second of these possibilities allows one to define a notion of smoothness, where a map is smooth if it maps smooth curves into smooth curves. We will denote by  $\CnMfd$, the category of smooth manifolds that arises in the convenient setting. Actually, as far as the definition of the configuration space goes, these two approaches are equivalent. This was already discussed in details in \cite{BFR}, for the case of a scalar field and the generalization to higher rank tensor is straightforward. Let $\Lor(M)$ denote the space of Lorentzian metrics on $M$. We can equip it with a partial order relation  $\prec$ defined by:
 \be\label{partial:order}
 g' \prec g\  \textrm{if}\  g'(X,X) \geq 0\  \textrm{implies}\  g(X,X) > 0\,,
 \ee
 i.e. the closed lightcone of $g'$ is contained in the lighcone of $g$. Note that, if $g$ is globally hyperbolic, then so is $g'$. 
 We are now ready to define a functor $\E:\Loc\rightarrow\LcMfd$ that assigns to a spacetime, the classical configuration space. To an object $\Mcal=(M,g_0)\in \obj(\Loc)$ we assign
 \be\label{config}
\E(\Mcal)\doteq \{g\in \Lor(M)|\,g\prec g_0\}\, .
 \ee
 Note that, if $g_0$ is globally hyperbolic, then so is $g\in\E(M,g_0)$. 
 The spacetime $(M,g)$ is also an object of $\Loc$, since it inherits the orientation and time-orientation from $(M,g_0)$. A subtle point is the choice of a topology on $\E(\Mcal)$. Let $\Gamma((T^*M)^{\otimes2})$ be the space of smooth contravariant 2-tensors. We equip it with the topology $\tau_W$, given by open neighborhoods of the form $U_{g,V}=\{g+h,h\in V\textrm{ open in }\Gamma_c((T^*M)^{\otimes2})\}$. It turns out that $\E(\Mcal)$ is an open subset of  $\Gamma((T^*M)^{\otimes2})$ with respect to $\tau_W$ (for details, see the Appendix \ref{class} and \cite{BFR}). The topology $\tau_W$ induces on $\E(\Mcal)$ a structure of an infinite dimensional manifold modeled on the locally convex vector space $\Gamma_c((T^*M)^{\otimes 2})$, of compactly supported contravariant 2-tensors. The coordinate
chart associated to $U_{g,V}$ is given by $\kappa_{g}(g + h) = h$. Clearly, the coordinate change map between two charts is affine, so $\E(\Mcal)$ is an affine manifold. It was shown in \cite{BFR} that $\tau_W$ induces on the configuration space also a smooth manifold structure, in the sense of the convenient calculus \cite{Michor}, so $\E$ becomes a contravariant functor from $\Loc$ to $\CnMfd$ where morphisms $\chi$ are mapped to pullbacks $\chi^*$.
 \subsection{Functionals}
Let us now proceed to the problem of defining observables of the theory. We first introduce functionals $F:\E(\Mcal)\to\RR$, which are smooth in the sense of the calculus on locally convex vector spaces  \cite{Ham,Neeb} (see Appendix \ref{class} for details). 
In particular, the definition of smoothness which we use implies that for all $g\in\E(\Mcal)$, $n\in\NN$, $F^{(n)}(g)\in \Gamma'((T^*M)^n)$, i.e. it is a distributional section with compact support. Later, beside functionals, we will also need vector fields on $\E(\Mcal)$.
Since the manifold structure of $\E(\Mcal)$ is affine, the tangent and cotangent bundles are trivial and are given by: $T \E(\Mcal) = \E(\Mcal) \times \Gamma_c((T^*M)^{\otimes 2})$, $T^* \E(\Mcal) = \E(\Mcal) \times \Gamma'_c((T^*M)^{\otimes 2})$. By a slight abuse of notation we denote the space  $\Gamma_c((T^*M)^{\otimes 2})$ by $\E_c(\Mcal)$. The assignment of 
 $\E_c(\Mcal)$ to $\Mcal$ is a covariant functor from $\Loc$ to $\Vect$ where morphisms $\chi$ are mapped to pushforwards $\chi_*$. Another covariant functor  between these categories is the functor $\D$ which associates to a manifold the space $\D(\Mcal)\doteq\Ci_0(M,\RR)$ of compactly supported functions.
 
An important property of a functional $F$ is its spacetime support. Here we introduce a more general definition than the one used in our previous works, since we don't want to rely on an additive structure of the space of configurations. To this end we need to introduce the notion of \textit{relative support}. Let $f_1,f_2$ be arbitrary functions between two sets $X$ and $Y$, then
 \[
 \relsupp (f_1,f_2)\doteq \overline{\{x\in X| f_1(x)\neq f_2(x)\}}\,.
 \]
 Now we can define the spacetime support of a functional on $\E(\Mcal)$:
 \begin{align}\label{support}
\supp\, F\doteq\{ & x\in M|\forall \text{ neighbourhoods }U\text{ of }x\ \exists h_1,h_2\in\E(M),\\ &  \relsupp(h_1,h_2)\subset U 
\text{ such that }F(h_1)\not= F(h_2)\}\ .\nonumber
\end{align}
Another crucial property is \textit{additivity}. 
\begin{df}
Let $h_1, h_2, h_3\in\E(\Mcal)$, such that $\relsupp(h_1,h_2)\cap\relsupp(h_1,h_3)=\varnothing$. By definition of the relative support we have $h_3\!\upharpoonright_{U}=h_2\!\upharpoonright_{U}$, where $U\doteq (\relsupp(h_1,h_2))^c\cap (\relsupp(h_1,h_3))^c$ and the superscript $c$ denotes the complement in $M$. We can therefore define a function $h$ by setting
\[
h= h_3\!\upharpoonright_{(\relsupp(h_1,h_2))^c},\quad h= h_2\!\upharpoonright_{(\relsupp(h_1,h_3))^c}\,,
\]
We say that $F$ is additive if 
\be\label{add}
F(h_1)=F(h_2)+F(h_3)-F(h)\quad \textrm{holds.}
\ee
\end{df}
A smooth compactly supported functional is called \textit{local} if it is additive and, for each $n$,
the wavefront set of $F^{(n)}(g)$ satisfies: $\textrm{WF}(F^{(k)}(g))\perp T\textrm{Diag}^k(M)$ with the thin diagonal $\textrm{Diag}^k(M)\doteq\left\{(x,\ldots,x)\in M^k:x\in M\right\}$. In particular $F^{(1)}(g)$ has to be  a smooth section for each fixed  $g$. From the additivity property follows that  $F^{(n)}(g)$ is supported on the thin diagonal.
 The space of compactly supported smooth local functions $F:\E(\Mcal)\to\RR$ is denoted by $\F_\loc(\Mcal)$. The algebraic completion of $\F_\loc(\Mcal)$ with respect to the pointwise product
\be\label{prod}
F\cdot G(g)=F(g)G(g)
\ee
is a commutative algebra $\F(\Mcal)$ consisting of sums of finite products of local functionals. We call it the algebra of multilocal functionals.
 $\F$ becomes a (covariant) functor by setting $\F\chi(F)=F\circ \E\chi$, i.e. $\F\chi(F)(g)=F(\chi^*g)$.
\subsection{Dynamics}
Dynamics is introduced by means of a generalized Lagrangian $L$ which is a natural transformation between the functor of test function spaces $\D$ and the functor $\F_\loc$ satisfying
\be\label{L:supp}
\supp(L_\Mcal(f))\subseteq \supp(f)\,,\qquad \forall\, \Mcal\in\obj(\Loc),f\in\D(\Mcal)\,,
\ee
and the additivity rule 
\be\label{L:add}
L_\Mcal(f_1+f_2+f_3)=L_\Mcal(f_1+f_2)-L_\Mcal(f_2)+L_\Mcal(f_2+f_3)\,,
\ee
for $f_1,f_2,f_3\in\D(\Mcal)$ and $\supp\,f_1\cap\supp\,f_3=\emptyset$.  
The action $S(L)$ is defined as an equivalence class of Lagrangians  \cite{BDF}, where two Lagrangians $L_1,L_2$ are called equivalent $L_1\sim L_2$  if
\be\label{equ}
\supp (L_{1,\Mcal}-L_{2,\Mcal})(f)\subset\supp\, df\,, 
\ee
for all spacetimes $\Mcal$ and all $f\in\D(\Mcal)$. In general relativity the dynamics is given by the Einstein-Hilbert Lagrangian:
\be\label{EH}
L^{\sst EH}_{\Mcal}(f)(g)\doteq \int R[g]f\,d\mu_{ g},\quad g\in\E(\Mcal)\,,
\ee
where we use the Planck units, so in particular the gravitational constant $G$ is set to 1.
\subsection{Diffeomorphism invariance}
In this subsection we discuss the symmetries of \eqref{EH}. 
As a natural transformation $L^{\sst EH}$ is an element of $\Nat(\Tens_c,\F)$,\footnote{Both $\Tens_c$ and $\F$ have to be treated as functors into the same category. In \cite{BFV} this category  is chosen to be $\Top$, the category of topological spaces, but in the present context it is more natural to include some notion of smoothness. A possible choice is the category of convenient vector spaces \cite{Michor}.} where 
$\Tens_c(\Mcal)\doteq\bigoplus_{k}\Tens^k_c(\Mcal)$ and $\Tens_c(\Mcal)$ is the space of smooth compactly supported sections of the vector bundle $\bigoplus_{m,l}(TM)^{\otimes m}\otimes(T^*M)^{\otimes l}$. The space $\Nat(\Tens_c,\F)$ is quite large, so, to understand the motivation for such an abstract setting, let us now discuss the physical interpretation of $\Nat(\Tens_c,\F)$. In \cite{FR} we argued that this space contains quantities which are identified with diffeomorphism invariant partial observables of general relativity, similar to the approach of \cite{Rovelli:2001bz,Dittrich:2005kc,Thiemann:2004wk}. Let $\Phi\in \Nat(\Tens_c,\F)$. A test tensor $f\in \Tens_c(\Mcal)$ corresponds to a concrete geometrical setting of an experiment, so we obtain a functional $\Phi_\Mcal(f)$, which depends covariantly on the geometrical data provided by $f$. We allow arbitrary tensors to be test objects, because we don't want to restrict \textit{a priori} possible experimental settings. A simple example of an experiment is the length measurement, studied in detail in \cite{Ohl}. 
\begin{exa}\label{length}
Let $S:[0,1]\rightarrow \RR^4$, $\lambda\mapsto s(\la)$ be a spacelike curve in Minkowski space $\Mcal=(\RR^4,\eta)$. For $g=\eta+h\in\E(\Mcal)$ the curve is still spacelike, and its length is 
\[
\La_{g}(S)\doteq \int_0^1 \sqrt{|g_{\mu\nu}(s)\dot{s}^\mu\dot{s}^\nu|} d\la\,.
\]
Here $\dot{s}^\mu$ is the tangent vector of $s$. We write it as $\dot{s}^\mu=\dot{s} e^\mu$, with $\eta_{\mu\nu} e^\mu e^\nu=-1$. Expanding the formula above in powers of $h$ results in
\[
\La_{g}(S)=\sum_{n=0}^\infty (-1)^n\binom{\tfrac{1}{2}}{n} \int_0^1h_{\mu_1\nu_1}(s)\dots h_{\mu_n\nu_n}(s)\dot{s}e^{\mu_1} e^{\nu_1} \dots e^{\mu_n} e^{\nu_n} d\la\,.
\]
Now, if we want to measure the length up to the $k$-th order, we have to consider a field
\[
\La_{\Mcal}(f_S)(h)=\int f_{S,0}^{\mu\nu}\eta_{\mu\nu} d^4x +\int f_{S,1}^{\mu\nu}h_{\mu\nu}d^4x +\ldots+\int f_{S,k}^{\mu_1\nu_1\dots \mu_k\nu_k}h_{\mu_1\nu_1}\dots h_{\mu_k\nu_k}  d^4x\,,
\]
where the curve, whose length we measure, is specified by the test tensor $f_S=(f_{S,0},\dots,f_{S,k})\in \Tens_c(\Mcal)$, which depends on the parameters of the curve in the following way:
\begin{align*}
 f^{\mu_1\nu_1\dots \mu_k\nu_k}_{S,k}(x)&=(-1)^k\binom{\tfrac{1}{2}}{k}\int_0^1\delta(x-s(\la))\dot{s}e^{\mu_1}e^{\nu_1}\dots e^{\mu_k}e^{\nu_k}d\la,\quad k\geq1 \,,\\
f^{\mu\nu}_{S,0}(x)&=-\int_0^1\delta(x-s(\la))\dot{s}e^{\mu}e^{\nu}d\la\,.
\end{align*}
\end{exa}
The framework of category theory, which we are using, allows us also to formulate the notion of locality in a simple manner. It was shown in \cite{BFR} that natural transformations $\Phi\in\Nat(\Tens_c,\F)$, which are additive in test tensors (condition \eqref{L:add}) and satisfy the support condition \eqref{L:supp}, correspond to local measurements, i.e. $\Phi_{\Mcal}(f)\in\F_\loc(\Mcal)$.
The condition for a family $(\Phi_{\Mcal})_{\Mcal\in\obj(\Loc)}$ to be a natural transformation reads
\[
\Phi_{\Mcal'}(\chi_*f)(h)=\Phi_{\Mcal}(f)(\chi^*h)\,,
\]
where $f\in\Tens_c(\Mcal)$, $h\in\E(\Mcal')$, $\chi:\Mcal\rightarrow\Mcal'$. Now we want to introduce a BV structure on natural transformations defined above. One possibility was proposed in \cite{FR}, where
 an associative, commutative product was defined as follows:
\be\label{ntprod}
(\Phi\Psi)_{\Mcal}(f_1,...,f_{p+q})=\frac{1}{p!q!}\sum\limits_{\pi\in P_{p+q}}\Phi_{\Mcal}(f_{\pi(1)},...,f_{\pi(p)})\Psi_{\Mcal}(f_{\pi(p+1)},...,f_{\pi(p+q)})\, .
\ee
Note, however, that the dependence on test tensors $f_i$ physically corresponds to a geometrical setup of an experiment, so $\Phi_{\Mcal}(f_1)\Psi_{\Mcal}(f_2)$ means that, on a spacetime $\Mcal$, we measure the observable $\Phi$ in a region defined by $f_1$ and $\Psi$ in the region defined by $f_2$. From  this point of view, there is no \textit{a priori} reason to consider products of fields which are symmetric in test functions. Therefore, we take here a different approach and replace the collection of natural transformations with another structure. Let us fix $\Mcal$.
We have already mentioned that the test function specifies the geometrical setup for an experiment, but a concrete choice of $f\in \D(\Mcal)$ can be made only if we fix some coordinate system\footnote{In general, it is more natural to work with a frame instead of a coordinate system, but we leave this problem for future study.}. This is related to the fact that, physically, points of spacetime have no meaning. To realize this in our formalism we have to allow for a freedom of changing the labeling of the points of spacetime. From now on we restrict the class of objects of $\Loc$ to spacetimes which admit a global coordinate system. 
Following ideas of Nakanishi \cite{Naka,Nak78} we realize the choice of a coordinate system by introducing four scalar fields $X^\mu$, which will parametrize points of spacetime. We can now consider the metric as a function of $X^\mu$, $\mu=0,\ldots,3$, i.e.  we write 
\[
g(x)=\sum\limits_{\nu,\mu}\vr{g}_{\mu\nu}(X(x))(dX^\mu  \otimes_s dX^\nu)(x)\,,
\] 
where  $\vr{g}$ is a function $\vr{g}:\RR^4\rightarrow\RR^{10}$, which represents $g\in\E(\Mcal)\doteq\Gamma((T^*M)^{\otimes2 })$ in the coordinate system induced by $X$, and we use the notation $g=X^*\vr{g}$. Let $\C(\Mcal)$ denote the space of global coordinate systems. We can write any test tensor  $f\in\Tens_c(\Mcal)$ in the coordinate basis induced by $X\in\C(\Mcal)$, so if we fix $\vr{f}\in\RR^k\rightarrow\RR^l$ for appropriate dimensions $k$ and $l$, then the change of $f=X^*\vr{f}$ due to the change of the coordinate system is realized through the change of scalar fields $X^\mu$. For a natural transformation $\Phi\in\Nat(\Tens_c,\F)$ we obtain a map
\[
\Phi_{\Mcal\vr{f}}(g,X)\doteq \Phi_\Mcal(X^*\vr{f})(g)\,,
\]
As long as $\Mcal$ is fixed, we will drop $\Mcal$ in $\Phi_{\Mcal\vr{f}}$ and use the notation $\Phi_{\vr{f}}$ instead.
The Einstein-Hilbert action induces a map
\[
L^{\sst EH}_{\vr{f}}(g,X)=\int_M R[g](x)\vr{f}(X(x))d\mu_{g}(x).
\]
For now we treat $g$ as a dynamical variable and $X^\mu$ are treated as external fields. Note that in the fixed coordinate system $X$ the components of $\vr{g}$ satisfy the condition:
\be\label{cond}
\frac{1}{\sqrt{-\vr{g}}}\tfrac{\partial}{\partial X^\mu}(\sqrt{-\vr{g}}\vr{g}^{\mu\nu})\circ X=\Box_{g}X^\mu\,,
\ee
%
Let us now consider the transformation of $g$ and $X$ under diffeomorphisms. Let $\al\in\Diff(M)$, then the transformed coordinate system is given by $X'(x)=X(\al(x))$ and the transformed $g$ is the pullback $\al^*g$. Infinitesimally, the transformation of the metric is given by the Lie derivative, so we define the action $\rho$ of the algebra $\X_c(\Mcal)\doteq\Gamma_c(TM)$ by
\be\label{var}
(\rho(\xi)\Phi_{\vr f})=\left<\frac{\delta \Phi_{\vr f}}{\delta g}\Big|_{X},\rho(\xi)g\right>+\left<\frac{\delta \Phi_{\vr{f}}}{\delta X^\mu}\Big|_{g},\pounds_\xi X^\mu\right>\,.
\ee
Note that in the coordinate system induced by $X$ we have $\pounds_\xi X^\beta=\vr{\xi}^\beta\circ X$, where $\vr{\xi}^\beta\circ X$ is understood as a scalar field.  
Diffeomorphism invariance of the Einstein-Hilbert Lagrangian means that
\[
\rho(\xi)L_{\vr{f}}^{\sst EH}=0\,,
\]
for $X^*\vr{f}\equiv 1$ on $\supp\,\xi$. Moreover, with this choice of  $\vr{f}$, also $\left<\frac{\delta L_{\vr{f}}^{\sst EH}}{\delta X}\Big|_{g},\pounds_\xi X\right>=0$, so we have two symmetries of the action:
\begin{align}
\rho_1(\xi)&=\left<\frac{\delta}{\delta g}\Big|_{X},\rho(\xi)g\right>\,,\label{sym1}\\
\rho_2(\xi)&=\left<\frac{\delta}{\delta X}\Big|_{g},\pounds_\xi X\right>\, .\label{sym2}
\end{align}
The first of these symmetries is a dynamical local symmetry and we will see later on that it causes the failure of the field equations to be normally hyperbolic. The other symmetry is non-dynamical and it involves variation with respect to the external fields $X^\mu$. Although the action is invariant under both of these symmetries, the diffeomorphism invariance of observables is the weaker requirement that functionals are invariant under the sum of these symmetries, i.e. they satisfy
\be\label{diffeo0}
\rho(\xi)\Phi_{\vr{f}}=0\,.
\ee
This corresponds exactly to the invariance condition for natural transformations, proposed in \cite{FR}, since the second term implements the action of infinitesimal diffeomorphisms on the test function. Our notion of diffeomorphism invariant objects is similar to the notion of gauge BRS invariant observables of gravity proposed by Nakanishi in \cite{Naka,Nak78} (see also \cite{NaOji}). The author makes there a distinction between the intrinsic BRS transformation and the total BRS transformation. The latter corresponds to our $\rho_1$, whereas the former corresponds to $\rho=\rho_1+\rho_2$, if one restricts oneself to test objects, which are scalar densities. In general the intrinsic BRS operator, as proposed by Nakanishi, has no geometrical meaning on the classical level and on the quantum level cannot be implemented by commutator with a local charge. Therefore, we do not follow this approach, but instead we make the coordinates $X$ dynamical. This is discussed in the next section.
\subsection{Metric-dependent coordinates}\label{metricdep}
Up to now we have considered the coordinates $X$ to be external fields independent of the metric. As a consequence, the diffeomorphism transformation \eqref{var} involves the term where variation with respect to $X^\mu$ is present. To avoid this, we can replace $X^\mu$ with some scalars $X_g^\mu$, $\mu=0,\ldots,3$, which depend locally on the metric. The particular choice of these fields is not relevant for the present discussion. They could be, for example, scalars constructed from the Riemann curvature tensor and its covariant derivatives (see  \cite{Khavkine:2015fwa}, which uses the earlier work of \cite{Berg,BergKom}). The caveat is that some particularly symmetric spacetimes do not admit such metric dependent coordinates, since in such cases the curvature might vanish (for a detailed discussion see \cite{CHP09,CH10}). This is however a non-generic case and in the situation where we are interested in, pure gravity without matter fields, such spacetimes are physically not observable. If matter fields are present, one can construct $X^\mu$'s using them. A known example is the Brown-Kucha\v{r} model \cite{BrK}, which uses dust fields. Here we briefly discuss a similar Ansatz, where the gravitational field is coupled to 4 scalar massless fields. We add to the Einstein-Hilbert action a term of the form
\[
L^{\sst KG}(f)(g,\phi^0,\dots,\phi^3)=\sum_{\al=0}^{3}\int_M (\nabla_g \phi^\al)^2 d\mu_{g}.
\]
The additional scalar fields satisfy the equations of motion
\[
\Box_g \phi^\al=0,\ \al=0,\dots,3\,.
\]
Classically, we can now identify the coordinate fields with the matter fields $\phi^\al$, i.e. we set $X_{g,\phi}^\mu=\phi^\mu$, $\mu=0,\dots,3$. With quantization in mind, we make the split of $g$ and $\phi^\al$ into background and perturbations, which will subsequently be treated as quantum fields. We set $g=g_0+\la h$ and $\phi^\al=\ph_0^\al+\la \ph^\al$. Our gauge-invariant observables are of the form
\[
\Phi_{\vr{f}}(h,\ph^0,\dots,\ph^3)=\Phi_{(M,g_0)}(\phi^*\vr{f})(\la h)\,,
\]
where $\phi^*\vr{f}(x)\doteq \vr{f}(\phi^0(x),\dots,\phi^3(x))$. As a concrete example consider
\[
\Phi_{\vr{f}}(h,\ph^0,\dots,\ph^3)=\int_M R_{\mu\nu\al\beta}R^{\mu\nu\al\beta}[g_0+\la h]\vr{f}((\ph_0^0+\la\ph^0)(x),\dots,(\ph_0^3+\la\ph^3)(x))d\mu_{g_0+\la h}\,,
\]
where $\ph_0^\al$ define harmonic coordinates with respect to the background metric, i.e. $\Box_{g_0}\ph_0^\al=0$, $\al=0,\dots 3$ and we choose $\vr{f}$ such that $\ph_0^*\vr{f}$ is compactly supported. The physical interpretation of the scalar fields $\phi^\al$ has to be made clear in concrete examples. We will come back to this problem in our future works.

On generic spacetimes matter fields are not necessary and it is enough to use the curvature scalars. Let us denote by $\beta$ the map $g\mapsto (X_g^0,\ldots,X_g^3)$ and we define 
\be\label{betaX}
\Phi^\beta_{\vr{f}}(g)\doteq \Phi_{\vr{f}}(g,X_g)\,.
\ee
Here we do not need to worry anymore if $X_g^\mu$ define an actual coordinate system or not, but we have to make sure that the support of $\vr{f}$ is contained in the interior of the image of $M$ inside $\M$ under the quadruple of maps  $X_g^\mu$, for all $g$ of interest. To ensure that, we restrict ourselves to a sufficiently small neighborhood $\Ocal\subset\E(\Mcal)$ of the reference metric $g_0$. This restriction is not going to be relevant later on, as quantisation is done perturbatively anyway.

Let $\Fcal(\Mcal)$ denote the algebra generated by functionals $\Phi^\beta_{\vr{f}}$ where $\vr{f}$ has compact support contained in the interior of $\bigcap_{g\in\Ocal} X_g(M)$. Note that elements of this space are no longer compactly supported in the sense of definition \eqref{support}, since the support of the functional derivative  $(\Phi^\beta_{\vr{f}})^{(1)}(g)$ can be different for different points $g\in\Ocal$, even though each  $(\Phi^\beta_{\vr{f}})^{(1)}(g)$ is a compactly supported distribution. They are also not local, because $X_g^*\vr{f}$ can depend on arbitrary high derivatives of the metric $g$. An advantage of using $\Fcal(\Mcal)$ is that the transformation law under diffeomorphisms takes a simpler form, namely
\[
\rho_1(\xi)\Phi^\beta_{\vr f}=(\rho(\xi)\Phi)^\beta_{\vr f}
\]
where $\rho=\rho_1+\rho_2$, as defined in \eqref{sym1} and \eqref{sym2}. To see this, note that
\begin{align*}
(\rho_1(\xi)(\Phi^\beta_{\vr{f}}))(g)&=\left<\frac{\delta \Phi^\beta_{\vr f}(g)}{\delta g}\Big|_{X},\pounds_cg\right>+\left<\frac{\delta \Phi^\beta_{\vr{f}}(g)}{\delta X^\mu}\Big|_{g},\pounds_c X_g^\mu\right>=\\
&=(\rho(\xi)\Phi_{\vr{f}})(X_g,g)=(\rho(\xi)\Phi)^\beta_{\vr f}
\end{align*}
This becomes particularly relevant for the construction of the BV differential $s$, which we will perform in the next section. In particular, as $\rho_2$ is not a dynamical symmetry, it cannot be implemented consistently within the BV formalism by means of the antibracket. From this reason, it is better to work on $\Fcal(\Mcal)$, where only $\rho_1$ is necessary.

The downside is the non-locality which we introduced by introducing the field dependent coordinates. This, however, is well under control, since the new complex is isomorphic to the old one. Besides, a non-local dependence on field configurations is necessary to obtain meaningful diffeomorphism invariant quantities, as we know that there are no local diffeomorphism invariant observables in general relativity.
\subsection{An abstract point of view on field dependent coordinates}\label{abstract:view}
More generally, there is no reason to distinguish between the curvature invariants that enter the definition of $X_g$'s and those which constitute the density $\Phi_x$ in  $\Phi^\beta_{\vr{f}}(g)=\int_M\Phi_x(g)\vr{f}(X_g(x))$. Abstractly speaking, one can consider a family of $N$ scalar  curvature invariants $R_1,\ldots, R_N$ and a class of globally hyperbolic spacetimes characterized by the 4-dimensional images under this $N$-tuple of maps. It was shown in \cite{MuSa} that any globally hyperbolic spacetime with a time function $\tau$ such that $|\nabla\tau|\ge1$, can be isometrically embedded into the $N$-dimensional Minkowski spacetime $\M^N$ for a sufficiently large $N$ (fixed by the spacetime dimension). This suggests that, depending on the physical situation, one can always choose $N$ and construct  $R_1,\ldots, R_N$  in such a way that all spacetimes of interest are characterized uniquely in this framework. One can then consider observables of the form
\[
\int_{M}\vr{f}(R_1(x),\ldots, R_N(x))\,,
\]
where $\vr{f}:\M^N\rightarrow\Omega^4(M)$ is a density-valued function, which we assume to be compactly supported inside the image of $M$ under the embedding $\ph:M\rightarrow \M^N$ defined by the family $R_1,\ldots,R_N$. One could then quantize the metric perturbation, in the same way as we do it in the present work. An alternative approach would be to quantize the embedding $\ph$ itself, as it was done for the bosonic string quantization in \cite{BRZ}. We hope to explore these possibilities in our future works.
\subsection{BV complex}\label{BVcomplex}
In this section and in the following ones we fix the spacetime $\Mcal$ and the map $\beta$, so we can simplify the notation and write $\Phi_{\vr{f}}$ istead of $\Phi^\beta_{\Mcal\vr{f}}$ if no confusion arises. In the first step we construct the Chevalley-Eilenberg complex corresponding to the action $\rho$ of $\X_c(\Mcal)$ on $\Fcal(\Mcal)$.  The Chevalley-Eilenberg differential is constructed by replacing components of the infinitesimal diffeomorphism in \eqref{diffeo0} by ghosts, i.e. evaluation functionals on $\X_c(\Mcal)$ defined by $c^\mu(x)(\xi)\doteq \xi^\mu(x)$. 
$\CEcal(\Mcal)$, the underlying algebra of the Chevalley-Eilenberg complex, is the graded subalgebra of $\Ci(\E(\Mcal),\La\X'(\Mcal))$, generated by elements of the form $\Phi_{\vr{f}}$, where $\Phi\in\Nat(\Tens_c,\CE)$ and $\CE(\Mcal)\doteq\Ci_{\ml}(\E(\Mcal),\Lambda\X'(\Mcal))$.
The Chevalley-Eilenberg differential $\gamma_{\sst CE}$ is defined by
\begin{align}
\gamma_{\sst CE}:\ &\CEcal^q(\Mcal)\rightarrow\CEcal^{q+1}(\Mcal)\,,\nonumber\\
(\gamma_{\sst CE}\, \omega)(\xi_0,\ldots, \xi_q)&\doteq\sum\limits_{i=0}^q(-1)^{i+q}\left<\tfrac{\delta}{\delta g}\big|_{X}(\omega(\xi_0,\ldots,\hat{\xi}_i,\ldots,\xi_q)),\pounds_{{\xi_i}}g\right>+\nonumber\\
&+\sum\limits_{i<j}(-1)^{i+j+q}(\omega(-[\xi_i,\xi_j],\ldots,\hat{\xi}_i,\ldots,\hat{\xi}_j,\ldots,\xi_q)\,,
\end{align}
where $\xi_0,\ldots, \xi_q\in\X(\Mcal)$. 
To see that $\gamma_{\sst CE}$ maps $\CEcal(\Mcal)$ to itself, we have to use the fact that symmetries act locally, so  $\gamma_{\sst CE}$ maps local functionals into local functionals and can be also lifted to a map on natural transformations and hence is also well defined on $\CEcal(\Mcal)$. By construction $\gamma_{\sst CE}$ is nilpotent and, comparing with \eqref{diffeo0}, we see that the 0-th cohomology of $\gamma_{\sst CE}$ is the space of diffeomorphism invariant elements of $\Fcal(\Mcal)$. 
Now we construct the Batalin-Vilkovisky complex, following the ideas of \cite{FR}. Note that $\CE(\Mcal)$ can be formally seen as the space of multilocal, compactly supported functions on a graded manifold $\overline{\E}(\Mcal)=\E(\Mcal)[0]\oplus\X(\Mcal)[1]$. The underlying graded algebra of the BV complex, is formally $\Ci_{\ml}(\Pi T^*\overline{\E}(\Mcal))$ the graded algebra of multilocal functions on the odd cotangent bundle\footnote{By $\Pi T^*\overline{\E}(\Mcal)$ we mean the graded manifold $\E(\Mcal)[0]\oplus\X(\Mcal)[1]\oplus \E_c'(\Mcal)[-1]\oplus\X'_c(\Mcal)[-2]$. The fact that the fiber consists of duals of spaces of compactly supported sections is consistent with our choice of the manifold structure on  $\E(\Mcal)[0]\oplus\frakg(\Mcal)[1]$, which is  induced by the topology $\tau_W$ introduced in section \ref{conf:space}.} of $\overline{\E}(\Mcal)$.
 We define $\BVcal(\Mcal)$ to be its graded subalgebra generated by covariant fields which arise as $\Phi_{\vr f}$ for $\Phi\in\Nat(\Tens_c,\BV)$ with
\be
\BV\doteq\Ci_{\ml}\big(\E,\La\E_c\widehat{\otimes}\La\C_c\widehat{\otimes}\La{\frakg}'\widehat{\otimes}S^\bullet \frakg_c\big)\,.\label{BVfix}
\ee
 The sequential completion  $\widehat{\otimes}$ of the algebraic tensor product is explained in details in \cite{FR} . We denote a field multiplet in $\overline{\E}(\Mcal)$ by $\ph$ and its components by $\ph^\al$, where the index $\al$ runs through all the metric and ghost indices. ``Monomial'' elements\footnote{The name \textit{monomial}, used after \cite{DF}, highlights the fact that these functions are homogeneous functions of field configurations.} of $\BV(\Mcal)$  can be written formally as
\be
\label{Polynom}
F=\int f_F(x_1, \dots ,x_{m})\Phi_{x_1}\!\dots\Phi_{x_k}  \tfrac{\delta^r}{\delta \ph(x_{k+1})} \dots   \tfrac{\delta^r}{\delta \ph(x_{m})}\,,
\ee
where $\Phi_{x_i}$ are evaluation functionals, the product denoted by the juxtaposition is the graded symmetric product of $\BV(\Mcal)$, $\tfrac{\delta^r}{\delta \ph(x_{i})}$ are right derivatives and we keep the summation over  the indices $\alpha$ implicit. Polynomials are sums of elements of the form \eqref{Polynom}, where $f_F$ is a distributional density with compact support contained in the product of partial diagonals. The WF set of $f_F$
has to be chosen in  such a way, that $F$ is multilocal. In the appropriate topology (more details may be found in \cite{FR}) polynomials \eqref{Polynom}  are dense in  $\BV(\Mcal)$. We identify the right functional derivatives $\tfrac{\delta^r}{\delta \ph^\al(x)}$, which differ from the left derivatives by the appropriate sign, with the so called \textit{antifields}, $\Phi_\al^\dgr(x)$\footnote{The choice of right derivatives at this point is just a convention and we use it in this work to simplify the signs.}. Functional derivatives with respect to odd variables and antifields are defined on polynomials as left derivatives and are extended to $\BV(\Mcal)$ by continuity.  In what follows, $\tfrac{\delta}{\delta \ph^\al(x)}$, $\tfrac{\delta}{\delta \ph_\al^\dgr(x)}$ denotes left derivatives. 

$\BVcal(\Mcal)$ is a graded algebra with two gradings:  the pure ghost number $\#\pg$ and the antifield number $\#\af$. Functionals on $\overline{\E}(\Mcal)$ have $\#\pg=0$, $\#\af=0$; ghosts have $\#\pg=1$ and  $\#\af=0$. Vector fields on $\overline{\E}(\Mcal)$ have the antifield number assigned according to the rule $\#\af(\Phi_\al^\ddagger(x))=\#\pg(\Phi^\al(x))+1$. We define the total grading of 
$\BVcal(\Mcal)$, the so called total ghost number by setting $\#\gh=\#\pg-\#\af$. 

Since $\BVcal(\Mcal)$ is the subalgebra of the algebra of functions on the odd cotangent bundle $\Pi T^*\overline{\E}(\Mcal)$, its elements are graded multivector fields and $\BVcal(\Mcal)$ carries a natural graded bracket $\{.,.\}$ (called the antibracket), which is defined as minus the usual Schouten bracket, i.e.
\[
\{F,G\}=\left<\frac{\delta^r F}{\delta \ph^\al},\frac{\delta^l G}{\delta \ph^\ddagger_\al}\right>-\left<\frac{\delta^r F}{\delta \ph^\ddagger_\al},\frac{\delta^l G}{\delta \ph^\al}\right>\,.
\]
Let us now discuss the field equations. Taking $\left<\frac{\delta}{\delta g}L^{\sst EH}_{\vr{f}}(g),h\right>$ and   choosing $\vr{f}$ such that $\vr{f}(X_g)\equiv 1$ on the support of $h$, we arrive at Einstein's equation in the vacuum:
\be\label{eom0}
R_{\mu\nu}[g]=0\,.
\ee
Let $\E_S(\Mcal)$ be the space of solutions to \eqref{eom0}. We are interested in characterizing the space of covariant fields on  $\E_S(\Mcal)$, which can be characterized as the quotient $\Fcal_S(\Mcal)=\Fcal(\Mcal)/\Fcal_0(\Mcal)$, where  $\Fcal_0(\Mcal)\subset\Fcal(\Mcal)$ is the ideal of $\Fcal(\Mcal)$ generated by the equations of motion, i.e. it is the image of the Koszul operator $\delta_{\sst EH}$ defined by
\be\label{Koszul}
\delta_{\sst EH} \Phi_{\vr{f}'}=\{ \Phi_{\vr{f}'},L^{\sst EH}_{\vr{f}}\},\ \Phi_{\vr{f}'}\in \BVcal(\Mcal),\,\vr{f}\equiv 1\ \textrm{on }\supp\, \vr{f}'\,,
\ee
To simplify the notation, we write from now on $\delta_{\sst EH}  \Phi_{\vr{f}'}=\{ \Phi_{\vr{f}'},S^{\sst EH}\}$ instead of  (\ref{Koszul}). In a similar manner, one can find a natural transformation $\theta^{\sst CE}$, that implements $\gamma_{\sst CE}^*$, i.e.  $\gamma_{\sst CE}^* =\{\ \cdot\ ,\theta^{\sst CE}\}$. For future convenience, we choose $\theta^{\sst CE}$ as
\be\label{choice:theta}
\theta^{\sst CE}_{\vr{f}}(g,c)=\left<\frac{\delta}{\delta g},\pounds_{fc}g\right>+\left<\frac{\delta^r}{\delta c},c^\mu\partial_\mu(fc)\right>\,,
\ee
where $f=X_g^*\vr{f}$. The motivation for the above form of $\theta^{\sst CE}_{\sst \Mcal}(f)$ is to introduce the cutoff for the gauge transformation by multiplying the gauge parameters with a compactly supported function $f$.
The total BV differential is the sum of the Koszul-Tate and the Chevalley-Eilenberg differentials:
 \[
 s_{\sst BV}\doteq\{\ \cdot\ ,S^{\sst EH}+\theta^{\sst CE}\}\,.
 \]
The nilpotency of  $s_{\sst BV}$ is guaranteed by the so called classical master equation (CME). In \cite{FR} it was formulated as a condition on the level of natural transformations. Here we can impose a stronger condition, with an appropriate choice of test functions. Let $\vr{f}\doteq(\vr{f}_1,\vr{f}_2)$ be a tuple of test functions chosen in such a way that $\vr{f}_i(X_g)$, $i=1,2$ is compactly supported 
for all $g\in\Ocal\subset\E(\Mcal)$ for an appropriately chosen small neighborhood $\Ocal$ of $g_0$. A pair of Lagrangians $(L^{\sst EH},\theta^{\sst CE})$, acts on the test functions according to
\be\label{pairing}
L^{\ex}_{\vr{f}}\doteq L^{\sst EH}_{\vr{f}_1}+\theta^{\sst CE}_{\vr{f}_2}\,,
\ee
For simplicity we will write just $L^{\sst EH}$ instead of $(L^{\sst EH},0)$, so $L^{\sst EH}_{\vr{f}}\equiv L^{\sst EH}_{\vr{f}_1}$, similarly for the other terms.

The choice of different test functions is motivated by the fact that they have slightly different meaning in our formalism and a different physical interpretation. The test function $\vr{f}_1$ is the cutoff for the Einstein-Hilbert interaction Lagrangian and $\vr{f}_2$ is used to multiply the gauge parameters in order to make the gauge transformations compactly supported. From this perspective, it is natural to require that $\vr{f}_1\equiv 1$ on the support of $\vr{f}_2$. This way, the gauge transformations doesn't see the cutoff of the theory. 

With an appropriate choice of a natural Lagrangian $\theta^{\sst CE}$ which generates $\gamma_{\sst CE}$ (as for example the one made in \eqref{choice:theta}), a stronger version of the {\cme} is fulfilled, namely
\be\label{CME}
\tfrac{1}{2}\{L^{\sst EH}_{\vr f}+\theta^{\sst CE}_{\vr f},L^{\sst EH}_{\vr f}+\theta^{\sst CE}_{\vr f}\}=0\,,
\ee
for any compactly supported $\vr{f}$, constructed as above.

Now, the fact the $\delta_{\sst EH}$ (graded-)commutes with $\gamma_{\sst CE}$ is the consequence of the invariance of the field equations under infinitesimal diffeomorphism. As $\delta_{\sst EH}^2=0=\gamma_{\sst CE}^2$, we conclude that $s_{\sst BV}^2=0$. A crucial feature of the BV formalism is the fact that the cohomology of the total differential can be expressed with the cohomology of $\gamma_{\sst CE}$ and the homology $\delta_{\sst EH}$. For this to hold $(\BVcal(\Mcal),\delta_{\sst EH})$ has to be a resolution (i.e. the $H_k$'s are trivial for $k<0$). To see this, we can look at the first row of the BV bicomplex with $\#\pg=0$. We have
 \[
\ldots\rightarrow\La^2\Vcal\oplus\Gcal\xrightarrow{\delta_{\sst EH}\oplus\rho}\Vcal\xrightarrow{\delta_{\sst EH}}\Fcal\rightarrow 0\,,
\]
where $\Vcal(\Mcal)$ is the subalgebra of $\BVcal(\Mcal)$ consisting of vector fields on $\E(\Mcal)$ and $\Gcal(\Mcal)$ is generated by elements of the form $\Phi_{\vr f}$ for $\Phi\in\Nat(\Tens_c,\fG)$, where $\fG(\Mcal)\doteq\Ci_\ml(\E(\Mcal),\X_c(\Mcal))$. Here $\rho$ is the map defined in \eqref{var}, so its image exhausts the kernel of $\delta_{\sst EH}$ and the sequence is exact in degree 1. This reasoning extends also to higher degrees, so one shows that the complex above is a resolution. The same argument can be repeated for all the rows of the BV bicomplex. Using standard methods of homological algebra, we can now conclude that the 0-th cohomology of $s_{\sst BV}$ on $\BVcal(\Mcal)$ is 
given by
\[
H^0(\BVcal(\Mcal), s_{\sst BV})=H^0((\BVcal(\Mcal),\delta_{\sst EH}), \gamma_{\sst CE})\,,
\]
and can be interpreted as $\BVcal^{\,ph}(\Mcal)$, the space of gauge invariant on-shell observables. 

In the next step we introduce the gauge fixing along the lines of \cite{FR}. For the specific choice of gauge we need, we have to extend the BV complex by adding auxiliary scalar fields: 4 scalar antighosts $\bar{c}_\mu$ in degree $-1$ and 4 scalar Nakanishi-Lautrup fields $b_\mu$, $\mu=0,...,3$ in degree $0$. The new extended configuration space is again denoted by $\overline{\E}(\Mcal)$ and the extended space of covariant fields on the new configuration space by $\BVcal(\Mcal)$. We define
\begin{align*}
s(\overline{c}_\mu)&=ib_\mu-\pounds_c\overline{c}_\mu\,,\\
s(b_\mu)&=\pounds_c b_\mu\,.
\end{align*}
To implement these new transformation laws we need to add to the Lagrangian a term
\[
\left<\frac{\delta^r}{\delta \overline{c}_\mu},if_2b_\mu-\pounds_{f_2c} \overline{c}_\mu,\right>+\left< \frac{\delta^r}{\delta b_\mu},\pounds_{f_2c} b_\mu,\right>\,,
\]
where $f_2=\vr{f}_2\circ X_g$

Next, we perform an automorphism $\al_\Psi$ of $(\BVcal(\Mcal),\{.,.\})$ such that the part of the transformed action which doesn't contain antifields has a well posed Cauchy problem. We define
\be\label{gfermion}
\alpha_\Psi(F)\doteq\sum_{n=0}^{\infty}\frac{1}{n!}\underbrace{\{\Psi_{\vr{f}'},\dots,\{\Psi_{\vr{f}'}}_n,F\}\dots\}\,,
\ee
where $X_g^*\vr{f}'\equiv 1$ on $\supp\,F$ and 
 \be\label{gff}
 \Psi_{\vr{f}'}=i\sum_{\mu,\nu}\int\!\!((\partial_\mu\bar{\vr{c}}_\nu\vr{g}^{\mu\nu}-\tfrac{1}{2}\vr{b}_\mu \bar{\vr{c}}_\nu\vr{\kappa}^{\mu\nu})\vr{f}')(X_g(x))d\mu_g(x)\,,
 \ee
 where $\vr{\kappa}$ is a non-degenerate 2-form on $\RR^4$. The explicit appearance of this form in the gauge fixing Fermion is related to the choice of a dual pairing for Nakanishi-Lautrup fields. This pairing is also used to define the embedding of $\overline{\E}_c$ into  $\overline{\E}'$. We will see in the next section that, as long as one uses consistently the same pairing, all essential structures are independent of this choice. 
 \begin{multline*}
\{\Psi_{\vr{f}'},L^{\ex}_{\vr{f}}\}=-\int\!\!(\partial_\mu(\vr{f}_2\vr{b}_\nu)\vr{g}^{\mu\nu}-\tfrac{1}{2}\vr{f}_2\vr{b}_\mu \vr{b}_\nu\vr{\kappa}^{\mu\nu})\sqrt{-\det\vr{g}})(X_g)d^4X+\\
+i\int (\partial_\mu \overline{\vr{c}}_\nu\sqrt{-\det\vr{g}}\vr{g}^{\mu\al}\partial_\al (\vr{f}_2\vr{c}^\nu))(X_g(x))d^4X\,,
\end{multline*}
which can be rewritten as
\[
\int\!\!\left(-\partial_\mu(\vr{f}_2\vr{b}_\nu)\vr{g}^{\mu\nu}\right)(X_g)d\mu_g+\int\!\!\left(\tfrac{1}{2}\vr{f}_2\vr{b}_\mu \vr{b}_\nu\right)(X_g)\kappa^{\mu\nu}d\mu_g+i\int\!\! f_2\Box_{\tilde{g}}\bar{c}_\nu C^\nu d\mu_{g}\,,
\]
where $C^{\mu}\doteq\pounds_c X_gg^{\mu}$, and $\kappa^{\mu\nu}$ is now a non-degenrate 2-form on $M$. In the coordinate system defined by $X$ we have $C^{\mu}=\vr{c}^{\mu}\circ X_g \equiv c^{\mu}$, so the scalar fields  $C^{\mu}$ coincide with the components of the ghost field $c\in\X(\Mcal)$. We denote the first term in the above formula by $L^{\sst GF}_{\vr{f}_2}$ and the second by $L^{\sst FP}_{\vr{f}_2}$ (gauge-fixing and Fadeev-Popov terms, respectively). The full transformed Lagrangian is given by:
\be\label{extended}
L^\ex_{\vr{f}}=L^{\sst EH}_{\vr{f}_1}+L^{\sst GF}_{\vr{f}_2}+L^{\sst FP}_{\vr{f_2}}+L^{\sst AF}_{\vr{f}_2}\,,
\ee
where $L^{\sst AF}_{\vr{f}_2}$ is the term containing antifields. The re-defined $L^\ex_{\vr{f}}$ also satisfies \eqref{CME}.

The variables of the theory (i.e. the components $\ph^\al$ of the multiplet $\ph\in\overline{\E}(\Mcal)$) are now: the metric $g\in\E(\Mcal)$, the Nakanishi-Lautrup fields $b_\mu$ and the antighosts $\bar{c}_\mu$, $\mu=0,\dots,3$ (scalar fields), ghosts $c\in \X(\Mcal)$. 
We introduce a new grading, called the total antifields number $\#\ta$. It is equal to 0 for functions on $\overline{\E}(\Mcal)$ and equal to 1 for all the vector fields on $\overline{\E}(\Mcal)$. 
New field equations are now equations for the full multiplet $\ph=(g,b_\mu,c,\bar{c}_\mu)$, $\mu=0,\dots,3$ and are derived from the $\#\ta=0$ term of $L^\ex$, denoted by $L$. The corresponding action $S(L)$ is called \textit{the gauge fixed action}. The $\al_\Psi$-transformed BV differential $s=\al_\Psi\circ s_{\sst BV}\circ \al_\Psi^{-1}$ is given by:
\[
s=\{\cdot,S^\ex\}=\gamma+\delta\,.
\]
The differential $\delta$ is the Koszul operator for the field equations derived from $S$ and
 $\gamma$ is the gauge-fixed BRST operator $\gamma$. The action of $\gamma$ on $\Fcal(\Mcal)$ and the evaluation functionals $b_
\mu$, $c$, $\bar{c}_\mu$ is summarized in the table below:
\begin{center}
{\setlength{\extrarowheight}{2.5pt}
\begin{tabular}{ll}
\toprule%
& $\gamma$\\\otoprule%
$\Phi_{\vr{f}}\in\Fcal$&$\left<\frac{\delta \Phi_{\vr{f}}}{\delta g},\pounds_{c}g\right>$\\
 $c$&$-\frac{1}{2}[c,c]$\\
 $b_\mu$& $\pounds_cb_\mu$\\
  $\bar{c}_\mu$& $ib-\pounds_c \overline{c}_\mu$\\\bottomrule
\end{tabular}}
\end{center}
The equations of motion expressed in the $X_g$ coordinate system are:
\begin{align}
R_{\lambda\nu}[\vr{g}]&=-i\partial_{\la}\overline{\vr{c}}_\al\,\partial_{\nu}\vr{c}^\al-\partial_{(\la}\vr{b}_{\nu)}\label{e1}\\%
\Box_{\vr{g}}\vr{c}^\mu&=0\label{e2}\\%
\Box_{\vr{g}}\overline{\vr{c}}_\mu&=0\label{e3}\\%
\tfrac{1}{\sqrt{-\det\vr{g}}}\partial_\mu(\sqrt{-\det\vr{g}}\vr{g}^{\mu\nu})(X_g)&=\vr{b}_\mu(X_g)\kappa^{\mu\nu}\label{e4}
\end{align}
where $\vr{g}$, $\vr{b}_\mu$, $\vr{c}^\mu$, $\overline{\vr{c}}_\mu$ have to be understood as evaluation functionals and not as field configurations. The last equation implies that
\be\label{substitution}
\Box_g X_g^\nu=b^\nu\,,
\ee
where $b^\nu\doteq (\vr{b}_\mu\kappa^{\mu\nu})\circ X_g$. The equation for $b^\mu$ is obtained by using the Bianchi identity satisfied by $R_{\lambda\nu}[\vr{g}]$ in equation \eqref{e1} and takes the form
\be\label{e5}
\Box_{\vr{g}}\vr{b}_\mu=0\,.
\ee
The gauge condition \eqref{e4} is the generalized harmonic gauge, studied in detail in \cite{FriedRen} (see also \cite{Fried} for a review). With this choice of a gauge the initial value problem for the multiplet $(g,b_\mu,c,\overline{c}_\mu)$ is well posed and the linearized equations become hyperbolic. It turns out that for $\Mcal=(M,g_0)$, the choice $\kappa^{\mu\nu}=g_0^{\mu\nu}$  is particularly convenient, so from now on we will continue with this choice.
Since $s=\delta+\gamma$ and $(\BVcal(\Mcal),\delta)$ is a resolution, the space of gauge invariant on-shell fields is recovered as the cohomology  $\Fcal^{\,\inv}_S(\Mcal)=H^0(s,\BVcal(\Mcal))=H^0(\gamma,H_0(\delta,\BVcal(\Mcal)))$.
\subsection{Peierls bracket}
We are finally ready to define the Peierls bracket. The system of equations considered in the previous section can be linearized by computing the second derivative of $L_{\vr{f}}$ and defining the Euler-Lagrange derivative $S''_\Mcal$
as a map from the extended configuration space to the space of vector-valued distrubutions (details about the functional analytic aspects of this construction can be found in \cite{Rej}) given by
\[
\left<(S''_\Mcal)_{\beta\al},\psi^\al_1\otimes \psi^\beta_2\right>\doteq \left<\frac{\delta^l}{\delta\ph^\beta}\frac{\delta^r}{\delta\ph^\alpha}L_{\vr{f}},\psi^\al_1\otimes \psi^\beta_2\right>\,,
\]
where $\psi_1\in\overline{\E}(\Mcal)$, $\psi_2\in\overline{\E}_c'(\Mcal)$ are field configuration multiplets and $X^*\vr{f}\equiv 1$ on the support of $\psi_2$. To simplify the sign convention, we use both the right and the left derivative.  For $\kappa=\vr{g}_0$, an explicit construction shows that the retarded and advanced propagators exist. We give formulas for these propagators in the next section, for the case of linearization around a particular background. Let $\Delta^{R/A}_{g}$ denote the propagators obtained by linearizing around the metric $g$. We define a Poisson (Peierls) bracket on  $\BVcal(\Mcal)$ by:
\[
\Pei{A}{B}(g,b_\mu,c,\overline{c}_\mu)\doteq\sum_{\al,\beta}\skal{\frac{\delta^l A}{\delta\ph^\al}}{\De_{g}^{\al\beta}\frac{\delta^r B}{\delta\ph^\beta}}(g,b_\mu,c,\overline{c}_\mu),\qquad \Delta_{g}=\Delta^A_{g}-\Delta^R_{g}\,.
\]
Note that
the support of $\Pei{A}{B}_{g}$ is contained in the support of $\Pei{A}{B}_{g_0}$, where $g_0$ is the reference metric in $\Mcal=(M,g_0)$. Hence, $\Pei{.}{.}$ is a well defined operation on $\BVcal(\Mcal)$, taking values in the space of smooth functionals on $\E(\Mcal)$. However, $\BVcal(\Mcal)$ is closed under $\Pei{.}{.}$. In order to obtain a Poisson algebra, one needs a suitable completion $\overline{\BVcal}(\Mcal)$, which we define in Appendix \ref{class}. Now we want to see if $\Pei{.}{.}$ is compatible with $s$. First, note that the image of $\delta$ is a Poisson ideal, so $\Pei{.}{.}_{g}$ is well defined on $H^0(\delta, \overline{\BVcal}(\Mcal))$. It remains to show that, on $H^0(\delta, \overline{\BVcal}(\Mcal))$, $\gamma$ is a derivation with respect to $\Pei{.}{.}_{g}$.  To prove it, we have to show that
\[
m\circ(\gamma\otimes 1+1\otimes\gamma)\circ\Ga'_{\De_{g}}=m\circ\Ga'_{\De_{g}}\circ(\gamma\otimes 1+1\otimes\gamma)\,,
\]
where
\[
\Ga'_{\De_{g}} \doteq \sum_{\al,\beta}\left<{\De_{g}}^{\al\beta}, \frac{\delta^l}{\delta\ph^\al}\otimes\frac{\delta^r}{\ph^\beta}\right>\,,
\]
After a short calculation, we obtain the following condition (compare with Prop. 2.3. of \cite{Rej13}):
\be\label{derivation}
(-1)^{|\sigma|}{K_{g}}^{\sigma}_{\ \beta}(x)\Delta_{g}^{\beta\al}(x,y)+{K_{g}}^{\al}_{\ \beta}(y)\Delta_{g}^{\sigma\beta}(x,y)=\gamma(\Delta_g^{\sigma\al})\,,
\ee
where $|\sigma|$ denotes $\#\gh(\ph^\sigma)$, while $K_{g}$ is defined by
\[
\gamma_{0g}\Phi_x^\al=\sum_\sigma{K_{g}}^{\al}_{\ \sigma}(x)\Phi_x^\sigma\equiv (K_{g}\Phi)^\al\,,
\]
and $\gamma_{0g}$ is the linearization of $\gamma$ around $g$. 
In a more compact notation we can write this condition as
\[
(-1)^{|\sigma|}({K_{g}}\circ\Delta_{g})^{\sigma\al}+(\Delta_{g}\circ K^\dagger_g)^{\sigma\al}=\gamma(\Delta_g^{\sigma\al})\,,
\]
where $ K^\dagger_g$ means taking the transpose of the operator-valued matrix and adjoints of its entries.

In \cite{Rej13} it was shown that this condition 
holds when $K$ is linear and the causal propagator doesn't depend on the fields. Here we give the proof of the general case. The gauge invariance of the action in the stronger form used in \eqref{CME} implies that
\[
\left<\frac{\delta^l L_{\vr{f}'}}{\delta \ph^\al},\theta^\al_{\vr{f}}\right>=0\,,
\]
where $\theta^\al_{\vr{f}}$ is the term in $\theta_{\vr{f}}$ which multiplies $\Phi_\al^\ddagger$. We can now apply on the both sides the differential operator $\left<(\De_g^R)^{\mu\beta}\circ\frac{\delta^l}{\delta\ph^{\beta}}\frac{\delta^r}{\delta\ph^{\kappa}},(\De_g^R)^{\kappa\nu}\right>$ and obtain
\begin{multline*}
\left<(\De_g^R)^{\mu\beta}\circ\left<\frac{\delta^l}{\delta\ph^{\beta}}\frac{\delta^l}{\delta\ph^\al}\frac{\delta^r}{\delta\ph^\kappa}L_{\vr{f}'},\theta_{\vr{f}}^\al\right>,(\De_g^R)^{\kappa\nu}\right>\\+\left<(\De_g^R)^{\mu\beta}\circ\left<\frac{\delta^l}{\delta\ph^{\beta}}\frac{\delta^l}{\delta\ph^\al}L_{\vr{f}'},\frac{\delta\theta_{\vr{f}}^\al}{\delta\ph^\kappa}\right>,(\De_g^R)^{\kappa\nu}\right>\\+\left<(\De_g^R)^{\mu\beta}\circ\left<\frac{\delta^l}{\delta\ph^\al}\frac{\delta^r}{\delta\ph^\kappa}L_{\vr{f}'},\frac{\delta\theta_{\vr{f}}^\al}{\delta\ph^\beta}\right>,(\De_g^R)^{\kappa\nu}\right>\\+\left<(\De_g^R)^{\mu\beta}\circ\left<\frac{\delta L_{\vr{f}'}}{\delta\ph^\al},\frac{\delta^r}{\delta\ph^\kappa}\frac{\delta^l}{\delta\ph^\beta}\theta_{\vr{f}}^\al\right>,(\De_g^R)^{\kappa\nu}\right>=0\ .
\end{multline*}
Setting $\vr{f}'\equiv1$ on the support of $\vr{f}$ we see that the last term is proportional to equations of motion, so we can ignore it. In the remaining terms we can make use of the fact that $\Delta^R_g$ is the Green's function for $S''_\Mcal$, so we finally obtain
\[
-\left<\frac{\delta \De_g^R}{\delta \ph^\al},\theta^\al_{\vr{f}}\right>+(-1)^{|\mu|}\frac{\delta\theta_{\vr{f}}^\mu}{\delta\ph^\kappa}\circ(\De_g^R)^{\kappa\nu}+(\De_g^R)^{\mu\beta}\circ\frac{\delta\theta_{\vr{f}}^\nu}{\delta\ph^\beta}\os0\,,
\]
where ``$\os$'' means ``modulo the terms that vanish on-shell'', i.e. modulo the image of $\delta$. The extra sign appears because we had to change one left derivative into a right derivative. The expression above is treated as an operator on $\overline{\E}_c(\Mcal)$ and if we choose $X^*\vr{f}\equiv1$ on the support of the argument, we arrive at
\[
\gamma(\De_g^R)\os(-1)^{|\sigma|}({K_{g}}\circ\Delta_{g})^{\sigma\al}+(\Delta_{g}\circ K^\dagger_g)^{\sigma\al}\,.
\]
The same argument can be applied to $\De_g^A$, so the identity \eqref{derivation} follows. We conclude that $\gamma$ is a derivation with respect to $\Pei{.}{.}_{g}$ modulo terms that vanish on the ideal generated by the full equations of motion, i.e. modulo the image of $\delta$. It follows that $\gamma$ is a derivation on $H^0(\delta, \overline{\BVcal}(\Mcal))$, hence $\Pei{.}{.}_{g}$ induces a Poisson bracket on $\overline{\Fcal}^{\,\inv}_S(\Mcal)\doteq H^0(s,\overline{\BVcal}(\Mcal))=H^0(\gamma,(H_0(\delta,\overline{\BVcal}(\Mcal)))$. This way we obtain a Poisson algebra $(\overline{\Fcal}^{\,\inv}_S(\Mcal),\Pei{.}{.}_{g})$, which we interpret as a classical algebra of observables in general relativity, for a particular choice of coordinates \eqref{betaX}.
\section{Quantization}\label{quantization}
\subsection{Outline of the approach}\label{outline}
In the previous section we defined the classical theory, now we want to quantize this structure. The usual prescription involving the star product cannot be applied to $\{.,.\}_{g}$, because acting iteratively with the functional differential operator $\left<{\Delta_{g}}^{\al\beta}, \frac{\delta^l}{\delta\ph^\al}\otimes\frac{\delta^r}{\ph^\beta}\right>$ involves also derivatives of $\Delta_{g}$. Therefore, from the point of view of quantization, it is convenient to split the gauge fixed action $S$ into a free part and the rest and quantize the free theory first. One can make this split by writing the Taylor expansion of $L_{\vr{f}}$ around a reference metric $g_0$, so $h=g-g_0$ is the perturbation. Later on, $h$ will be interpreted as a quantum fluctuation around a classical background. Interaction is introduced in the second step, with the use of time-ordered products.

 To keep track of the order in $h$, it is convenient to introduce a formal parameter $\la$ (identified with the square root of the gravitational coupling constant, i.e. $\la=\sqrt{\kappa}$) and the field multiplet $(g_0+\la h,\la b, \la c, \la \bar{c})$, together with corresponding antifields $(\la h^\dagger,\la b^\dagger, \la c^\dagger, \la \bar{c}^\dagger)$. We denote $(h,b, c, \bar{c})$ collectively by $\ph$. We also redefine the antifields using the prescription $\ph_\al^\ddagger\mapsto \la \ph_\al^\ddagger$. It is convenient to use the natural units, where $\kappa$ is not put to 1, but has a dimension of length squared, so $h$ has a dimension of $1/\textrm{length}$. The action used in quantization must be dimensionless, so, as in the path integral approach, we use $L/\la^2$, where $L$ is the full extended action defined before. We denote 
 \[
 L_0\doteq \la L^{(1)}_{\sst(M,g_0)}(g_0,0,0,0)+\frac{\la^2}{2}{L^{(2)}}_{\sst(M,g_0)}(g_0,0,0,0))\]
 and consider it to be the free action. If $g_0$ is not a solution to Einstein's equations, the linear term doesn't vanish and the free equation of motion becomes a differential equation with a source term. Also, negative powers of $\la$ appear in the action. Formally, we can solve this problem by introducing another parameter $\mu$, so that $\tfrac{1}{\la}L^{(1)}_{\sst(M,g_0)}(g_0,0,0,0)\equiv \mu J_{g_0}$, where $J_{g_0}$ is the source term, linear in $h$. Our observables will now be formal power series in both $\la$ and $\mu$. For the physical interpretation we will restrict ourselves to spacetimes where $g_0$ is a solution and put $\mu=0$, but algebraically we can perform our construction of quantum theory on arbitrary backgrounds. 
 
We introduce the notation $S_I=S^\ex-S_0$ and $\theta=S^\ex-S$. We also expand $\theta$ around $g_0$. The first nontrivial term in the expansion is linear in configuration fields and we denote it by  $\theta_0$. It generates the free gauge-fixed BRST differential $\gamma_0$. The Taylor expansion of the classical master equation \eqref{CME} yields in particular:
\[
\{\theta_0,S^{(2)}(0)\}+\{\theta_0,\theta_0\}+\{\theta_1,S^{(1)}(0)\}\sim0\,.
\]
The first two terms of this identity correspond to the classical master equation for the free Lagrangian $S^{(2)}(0)+\theta_0$. The third term vanishes only for on-shell backgrounds, so $\gamma_0$ is a symmetry of $S_0$ only if $g_0$ is a solution of  Einstein's equations. Consequences of this fact are discussed in detail in \cite{Rej13}.

Observables are formal power series in $\la$ obtained by expanding elements of $\BVcal(\Mcal)$ around $(g_0,0,0,0)$. From now on $\BVcal(\Mcal)$ is implicitly understood as the space of formal power series in $\la$ and $\mu$. As a simple example consider the scalar curvature $R$ on an on-shell background $(M,g_0)$.
\begin{multline*}
\Phi_{\vr{f}}(g)=\int_M R[g_0]\vr{f}(X_{g_0})d\mu_{g_0}\\
\qquad\qquad\quad+\la\left( \int_M \vr{f}(X_{g_0})\left.\frac{\delta}{\delta g}(Rd\mu)\right|_{g_0}(h)+\int_M R[g_0]\partial_\mu \vr{f}(X_{g_0}) \left.\frac{\delta X^\mu_g}{\delta g}\right|_{g_0}(h)\right)+\Ocal(\la^2)\,,
\end{multline*}
where $\vr{f}$ is a compactly supported function on $\RR^4$, with the support inside the interior of the image of $M$ under $X_g$. Note that we do not need to make any restrictions on $h$ now, as our construction is perturbative and the choice of $\vr{f}$ refers only to the background metric $g_0$. Therefore, from now on we will consider the configuration space to be $\E(\Mcal)=\Gamma((T^*M)^{\otimes 2})$.
 
 %


Let us now summarize the general strategy for the perturbative quantization of gravity, which we will follow in this work. We start with the full classical theory, described by the gauge-fixed action $S$ which is invariant under the BRST operator $\gamma$. Then, we linearize 
the action and the BRST differential around a fixed background metric $g_0$. This way, the ``gauge'' invariance of the theory is broken and the linearized classical theory doesn't posses the full symmetry anymore. If we linearize around $g_0$ which is a solution of the full Einstein's equations, then part of the symmetry remains and $S_0$ is invariant under $\gamma_0$. This, however, is not needed for performing a deformation quantization of the linearized theory along the lines of  \cite{FR3}, which works for arbitrary $(M,g_0)\in\obj(\Loc)$. The free theory, quantized this way, still contains non-physical fields and is not invariant under the full BRST symmetry. 
This is to be expected, since the linearization breaks this symmetry in an explicit way. To restore the symmetry we have to include the interaction. This can be done with the use of time-ordered products and relative S-matrices. The full interacting theory is again invariant under the full BRST symmetry $\gamma$. This is guaranteed by the so called \textit{quantum master equation} (QME), which is a renormalization condition for the time-ordered products (see \cite{FR3} for more details). A crucial step in our construction is to prove that the quantized interacting theory which we obtain in the end doesn't depend on the choice of the background $g_0$. This will be done in section \ref{background}.
\subsection{Perturbative formulation of the classical theory}\label{free:theory}
 The starting point for the construction of the linearized classical theory is the gauge-fixed free action $S_0$. For simplicity we choose from now on the gauge with $\kappa=\vr{g}_0$. To write $S_0$ in a more convenient way, we introduce some notation. 
Let us define the divergence operator, which acts on symmetric covariant tensors  $\di:\Gamma(S^2T^*M)\rightarrow \Gamma(T^*M)$ by
 \[
(\di\, t)_\al\doteq \frac{1}{\sqrt{-\det g_0}}g_0^{\beta\mu}\partial_\mu (t_{\beta\al}\sqrt{-\det g_0})\,.
\]
Let us also introduce a product 
\[\left<u,v\right>_{g_0}=\int_M\left<u^{\#},v\right> d\mu_{g_0}\,,
\]
where $u,v$ are tensors of the same rank and $\#$ is the isomorphism between $T^*M$ and $TM$ induced by $g_0$. 
 The formal adjoint of $\di$ with respect to the product $\left<.,.\right>_{g_0}$ is denoted by $\ds:\Gamma(T^*M)\rightarrow \Gamma(S^2T^*M)$. 
 In local coordinates (in our case fixed by the choice of $X_{g_0}^\mu$ ) we obtain:
\[
(\ds v)_{\al\bet}=\frac{1}{2}(\partial_\beta v_\al+\partial_\al v_\bet)\,.
\]
Another important operation is the trace reversal operator $G:(TM)^{\otimes2}\rightarrow (TM)^{\otimes2}$, defined by
\be\label{trace:rev}
Gt=t-\frac{1}{2}(\tr t)g_0\,.
\ee
We have $\tr(Gt)=-\tr t$ and $G^2=\id$. 
Using this notation we can write the quadratic part of the gauge fixed Lagrangian on a generic background $\Mcal=(M,g)\in\obj(\Loc)$ in the form:
\[
{L_0}_{\vr{f}}=\int\limits_M \left.\frac{\delta}{\delta g}(Rfd\mu)\right|_{g_0}(h)+2i\sum_{\nu=0}^3 \big<d{\bar{c}}_\nu,d(f {c}^\nu)\big>_{{g}_0}+\big<fb,\di(G{h})-\tfrac{1}{2} {b}\big>_{{g}_0}\,,
\]
where $\left.\frac{\delta}{\delta g}(Rd\mu)\right|_{g_0}(h)$ denotes the linearization of the Einstein-Hilbert Lagrangian density around the background $g_0$ and ${b}$ is a 1-form on $M$ defined by $b\doteq\sum_\nu\vr{b}_\nu(X_{g_0}) dX_{g_0}^\nu$. Now we calculate the variation of ${L_0}_{\vr{f}}$, to obtain $S_{\sst\Mcal}''(x,y)$. We write it here in a block matrix form:
\begin{equation}\label{EOMs}
S_{\sst\Mcal}''(z,x)=\delta(z,x)\left(\begin{array}{cccc}
-\frac{1}{2}\left(\Box_LG +2G\ds\circ\di\circ G\right)&G\circ \di^*&0&0\\
\di\circ G&-1&0&0\\
0&0&0&-i\Box_H\\
0&0&i\Box_H&0
\end{array}\right)(x)\,,
\end{equation}
where the variables are $(h,b,c^0,...,c^3,\overline{c}_0,...,\overline{c}_3)$. In the formula above $\Box_H=\delta d$ is the Hodge Laplacian, $\delta\doteq\hinv d*$ is the codifferential and $\Box_L$ is given in local coordinates by
\be
(\Box_L h)_{\al\bet}= \nabla^\mu\nabla_\mu h_{\al\bet}-2(R_{(\al}^{\phantom{(\al}\mu}h^{\phantom{\mu}}_{\bet)\mu}+R_{(\al\phantom{\mu\nu}\bet)}^{\phantom{(\al}\mu\nu\phantom{\bet)}}h_{\mu\nu})\label{LL1}\,.
\ee
In the literature, $\Box_L$ it is called the Lichnerowicz Laplacian \cite{Lich} and it provides a generalization of the Hodge Laplacian to the space of symmetric contravariant 2 tensors. Note that $\Box_L$ commutes with $G$ on $\E(\Mcal)$. 
It is now easy to check that the retarded and advanced propagators for $S_0$ are given by:
\[
\Delta^{A/R}(x,y)=-2\left(\begin{array}{cccc}
G\Delta_t^{A/R}&G\Delta_t^{A/R}G\circ \ds_y&0&0\\
\di_x\circ \Delta_t^{A/R}&\di_x\circ \Delta_t^{A/R}G\circ \ds_y+\frac{1}{2}\delta_4&0&0\\
0&0&0&-i\Delta^{A/R}_s\\
0&0&i\Delta^{A/R}_s&0
\end{array}\right)\,,
\]
where $\delta_4$ denotes the Dirac delta in 4 dimensions and subscripts in $\di_x$ and $\di^*_y$ mean that the operator should be applied to the first, respectively, to the second variable. In the above formula $\Delta_t^{A/R}$ are the advanced/retarded propagators for the operator $\Box_L$ acting on symmetric tensor fields with compact support $\E_c(\Mcal)=\Gamma_c(S^2T^*M)$. Analogously, $\Delta_s^{A/R}$ are the propagators for $\Box_H$ on 0-forms (scalar functions). Using the above formula we can write down the expression for the causal propagator and use this propagator to define the classical linearized theory, by introducing the Peierls bracket:
\begin{equation*}
\Pei{F}{G}_{g_0} = \sum_{\al,\beta} \skal{\frac{\delta^l F}{\delta\ph^\al}}{{\De}^{\al\beta}\frac{\delta^r G}{\delta\ph^\beta}}\,,
\end{equation*}
where $\Delta=\Delta^R-\Delta^A$. Let us define microcausal functionals as smooth, compactly supported functionals whose derivatives (with respect to both $\ph$ and $\ph^\dgr$) satisfy the WF set condition:
 \be\label{mlsc}
\WF(F^{(n)}(\ph,\ph^\dgr))\subset \Xi_n,\quad\forall n\in\NN,\ \forall\ph\in\overline{\E}(\Mcal)\,,
\ee
where $\Xi_n$ is an open cone defined as 
\be\label{cone}
\Xi_n\doteq T^*M^n\setminus\{(x_1,\dots,x_n;k_1,\dots,k_n)| (k_1,\dots,k_n)\in (\overline{V}_+^n \cup \overline{V}_-^n)_{(x_1,\dots,x_n)}\}\,,
\ee
where $\overline{V}_\pm$ is the closed future/past lightcone with respect to the metric $g_0$. Let $\BV_\mc(\Mcal)$ denote the space of microcausal functionals. It is equipped with the H\"ormander topology $\tau_{\Xi}$, which allows to control properties of functional derivatives (see \cite{FR} for a precise definition). We extend the space of covariant fields to ones induced by natural transformations in $\Phi\in\Nat(\Tens_c,\BV_{\mc})$ and the algebra generated by the corresponding functionals $\Phi_{\vr{f}}^\beta$ is denoted by $\BVcal_{\mc}(\Mcal)$.
%
\subsection{Free quantum theory}\label{free}
In the next step we want to construct the quantized algebra of free fields by means of deformation quantization of the classical algebra $(\BVcal_\mc(\Mcal),\Pei{.}{.}_{g_0})$. To this end, we equip the space of formal power series $\BVcal_{\mc}(\Mcal)[[\hbar]]$ with a noncommutative star product. In this construction one needs Hadamard parametrices, i.e. a set of distributions in $\Dcal'(M^2)$ which fulfill
\begin{IEEEeqnarray}{rCl}\label{parametrix}
  \omega^{\alpha \beta}(x,y) - (-1)^{|\ph^\alpha| |\ph^\beta|} \omega^{\beta \alpha}(y,x)& =& i \Pei{\ph^\alpha(x)}{\ph^\beta(y)}_{g_0},\IEEEyessubnumber\label{classical:limit}\\
 \sum\nolimits_\bet O^\alpha_\beta \omega^{\beta \gamma} & =& 0\ \textrm{mod }\Ci\textrm{ function},\IEEEyessubnumber\label{field:eq}\\
 \WF(\omega^{\alpha \beta}) & \subset &C_+,\IEEEyessubnumber\label{WFset}\\
 \overline{\omega^{\alpha \beta}(x,y)} & =& \omega^{\beta \alpha}(y,x).\IEEEyessubnumber\label{hermitian}
\end{IEEEeqnarray}
Here $O^\alpha_\beta$ are the coefficients of the differential operator induced by $S''_{\sst\Mcal}$, written in the basis $\{ \ph^\alpha \}$. They can be easily read off from  (\ref{EOMs}). 
By $ C_+$ we denoted the following subset of the cotangent bundle $ T^*M^2$:
\[
 C_+ = \{ (x_1, x_2; k_1, - k_2) \in T^*M^2 \setminus \{ 0 \} | (x_1; k_1) \sim (x_2; k_2), k_1 \in \overline{V}^+_{x_1} \},
\]
where $(x_1; k_1) \sim (x_2; k_2)$ if there is a lightlike geodesic from $x_1$ to $x_2$ and $k_2$ is a parallel transport of $k_1$ along this geodesics.
 These are the properties which we will require for a Hadamard parametrix on the general background $\Mcal\in\obj(\Loc)$. 
If we replace the condition (\ref{field:eq}) by a stronger one 
 \begin{equation}\label{field:eq:s}
\sum_\bet O^\alpha_\beta \omega^{\beta \gamma} =0\,,
 \end{equation}
 then the Hadamard parametrix becomes a Hadamard 2-point function.  
We will now show that such a distribution can be constructed on generic backgrounds. Assume that $\omega$ is of the form:
 \begin{equation}\label{two:point}
\omega=-2\left(\begin{array}{cccc}
G\omega_t&\omega_t^T\, \ds_y&0&0\\
\di_x\, \omega_t&\di_x\, G\, \omega^T\, \ds_y&0&0\\
0&0&0&-i\omega_v\\
0&0&i\omega_v&0
\end{array}\right)\,,
\end{equation}
In this case, the conditions for $\omega$ to be a Hadamard 2-point function reduce to:
\begin{IEEEeqnarray}{rCl}\label{parametrix2}
  \omega_{v/t}(x,y) -\omega_{v/t}(y,x)& =& i\Delta_{v/t}(x,y),\IEEEyessubnumber\label{classical:limit2}\\
\Box_L\, \omega_t& =& 0,\ \Box_H\, \omega_v=0,\IEEEyessubnumber\label{field:eq2}\\
 \WF(\omega_{v/t}) & \subset &C_+,\IEEEyessubnumber\label{WFset2}\\
 \overline{\omega_{v/t}(x,y)} & =& \omega_{v/t}(y,x).\IEEEyessubnumber\label{hermitian2}
\end{IEEEeqnarray}
The existence of a Hadamard parametrix is already clear, since one just needs to pick arbitrary parametrices $ \omega_t$, $ \omega_v$ of $\Box_L$ and $\Box_H$ respectively. Their existence was already proven in \cite{SahlVer} (the paper actually discusses general wave operators acting on vector-valued field configurations). Now, from a parametrix, one can construct a bisolution using a following argument: let $\omega$ be a Hadamard parametrix and by $O$ we denote the hyperbolic operator from \eqref{field:eq}, so $O_x \omega=h,\ O_y\omega=k,$ hold
for some smooth functions  $h$ and $k$. Let $\chi$ be a smooth function such that $\supp\chi$ is past-compact and $\supp(1-\chi)$ is future-compact. Define
$$G_{\chi}\doteq \Delta^R\chi+\Delta^A(1-\chi)\,.$$
Clearly $G_{\chi}$ is a right inverse for $O$. A Hadamard bisolution $\omega_\chi$ can be now obtained as
$$\omega_\chi\doteq(1-G_{\chi}O)\circ \omega\circ (1-OG^T_{\chi})\,.$$
With the use of Hadamard 2-point functions and parametrices one can define on $\BVcal_{\mc}(\Mcal)[\hbar]]$ a noncommutative star product. To separate the functional analytic aspects of the framework from the algebraic structure, it is convenient to introduce the space of regular functionals $\BV_\reg(\Mcal)$, which is defined as the space of smooth functionals satisfying $\WF(F^{(n)}(\ph,\ph^\dgr))=\varnothing$ for all $\ph$, $\ph^\ddagger$, so their derivatives are compactly supported smooth functions. Here, in contrast to our previous works, we do not assume that these functionals are compactly supported.

We can define on $\BV_\reg(\Mcal)$ the star product $\star$, which provides the deformation quantization of $(\BV_\reg(\Mcal),\Pei{.}{.}_{g_0})$  as:
\begin{equation*}
F\star G\doteq m\circ \exp({i\hbar \Gamma'_\Delta})(F\otimes G),
\end{equation*}
where  $\Gamma'_\Delta$ is the functional differential operator
\begin{equation*}
\Gamma'_\Delta\doteq  \sum_{\al, \beta} \left<{\Delta}^{\al\beta},\frac{\delta^l}{\delta\ph^\al} \otimes \frac{\delta^r}{\delta\ph^\beta}\right>\,.
\end{equation*}

There is however, a problem with extending this structure to $\BVcal_\mc(\Mcal)$, due to the singularity structure of the causal propagator. To solve this problem, we replace $\Delta$ by a Hadamard 2-point function $\omega=\frac{i}{2}\Delta+H$. The resulting star product is given by
\begin{equation*}
F\star_H G\doteq m\circ \exp({i\hbar \Gamma'_\omega})(F\otimes G)\,.
\end{equation*}

The two star products introduced above provide isomorphic structures on $\BV_\reg(\Mcal)[[\hbar]]$ and this isomorphism is given by the map $\al_H\doteq e^{\frac{\hbar}{2}\Gamma_H}: \BV_\reg(\Mcal)[[\hbar]]\rightarrow \BV_\reg(\Mcal)[[\hbar]]$, where
\[
\Ga_{H} \doteq \sum_{\al,\beta}\left<{H}^{\al\beta}, \frac{\delta^l}{\delta\ph^\al}\frac{\delta^r}{\ph^\beta}\right>\,.
\]
Now, the star product $\star_H$ can be extended to $\BVcal_\mc(\Mcal[[\hbar]])$ and the resulting algebra is denoted by $\fA_H(\Mcal)$. Note that $\BV_\reg(\Mcal)[[\hbar]]$ is dense in $\BVcal_\mc(\Mcal[[\hbar]])$, if we equip it with the H{\"o}rmander topology. 
We can, therefore, use the intertwining map $\al_H:\BV_\reg(\Mcal)[[\hbar]]\rightarrow \BV_\mc(\Mcal)[[\hbar]]$ to define a certain ``completion'' of the source space $\BV_\reg(\Mcal)$ by extending  $\BV_\reg(\Mcal)$ with all elements of the form $\lim_{n\rightarrow \infty}\al_H^{-1}(F_n)$, where $(F_n)$ is a convergent sequence in $\BVcal_\mc(\Mcal)$ with respect to the H\"ormander topology. The resulting space, denoted by $\al_H^{-1}(\BVcal_\mc(\Mcal))$, is equipped with a unique continuous star product equivalent to $\star_H$,
\[
\al_H^{-1}F\star \al_H^{-1}G\doteq \al_H^{-1}(F\star_H G)\,.
\]
Different choices of $H$ differ only by a smooth function, hence all the algebras  $(\al_H^{-1}(\BVcal_\mc(\Mcal)[[\hbar]]),\star)$ are isomorphic and define an abstract algebra $\fA(\Mcal)$. For $F \in \fA(\Mcal)$,  we have $\al_HF\in \fA_H(\Mcal)$, hence we can realize $\fA(\Mcal)$ more concretely as the space of families $\{ G_H \}_H$,  labeled by possible choices of $H$, fulfilling the relation
\[
 G_{H'} = \exp(\hbar \Gamma_{H'-H}) G_H\,,
\]
equipped with the product
\[
 (F \star G)_H = F_H \star_H G_H.
\]
The support of $F \in \fA(\Mcal)$ is defined as $\supp(F) = \supp(\al_HF)$. Again, this is indepedent of $H$. Functional derivatives are understood as
$
\skal{\frac{\delta F}{\delta \ph}}{\psi} = \al_H^{-1}\skal{\frac{\delta \al_HF}{\delta \ph}}{\psi}\,,
$
which is well defined as $\Gamma_{H'-H}$ commutes with functional derivatives.
 
Polynomial functionals in $\fA_H(\Mcal)$ are interpreted as Wick powers.
Corresponding elements of $\fA(\Mcal)$ are obtained by applying $\al_H^{-1}$. 
The resulting object is denoted by
\be\label{polynomials1}
\int :\Phi_{x_1}\dots\Phi_{x_n}:_H f(x_1,\dots, x_n)\doteq \al^{-1}_H\Big(\int \Phi_{x_1}\dots\Phi_{x_n} f(x_1,\dots, x_n)\Big)\, ,
\ee
where $f\in\Ecal'_{\Xi_n}(M^n,V)$ and we suppress all the indices. 
Let us now discuss the covariance properties of Wick powers. The assignment of $\fA(\Mcal)$ to a spacetime $\Mcal$ can be made into a functor $\fA$ from the category $\Loc$ of spacetimes to the category of  topological *-algebras $\Obs$ and, by composing with a forgetful functor, to the category $\Vect$ of topological vector spaces. Admissible embeddings are mapped to pullbacks, i.e. for $\chi:\Mcal\rightarrow \Mcal'$ we set $\fA\chi F(\ph)\doteq F(\chi^*\ph)$. Locally covariant quantum fields are natural transformations between $\D$ and $\fA$.  
We require Wick products to be locally covariant in the above sense. Let  $\BVcal_\loc(\Mcal)$ denote the subspace of $\BVcal_\mc(\Mcal)$ generated 
(as a vector space) by natural transformations $\Nat(\Tens_c,\F_{\loc})$. Note that elements are local in a weaker sense, as the coordinates in $\Phi_\Mcal(X_g^*\vr{f})$ depend on the metric (albeit locally).

Let us now define covariant Wick products. On each object $\Mcal$ we have to construct the map ${\TT_1}_\Mcal$ from $\BVcal_\loc(\Mcal)$ (the ``classical world'') to the quantum algebra $\fA(\Mcal)$ in such a way that
\be
\label{covariance}
{\TT_1}_{\Mcal}(\Phi^{\beta}_{\Mcal\vr{f}})(\chi^*g)={\TT_1}_{\Mcal'}(\Phi^{\beta}_{\Mcal'\vr{f}})(g)\,,
\ee
As we have noted before, classical functionals can be mapped  to $\fA_H(\Mcal)$ by identification \eqref{polynomials1}. This, however, doesn't have the right covariance properties and \eqref{covariance}  would not be fulfilled. A detailed discussion of the analogous problem in the scalar field theory is presented in the section 5 of \cite{BFV}, 
where it is shown that redefining Wick products to become covariant amounts to solving a certain cohomological problem. The result reproduces the solution, which was proposed earlier in \cite{HW}. One has to define $\TT_1$ as $\al^{-1}_{H-w}$, where $w$ is the smooth part of the Hadamard 2-point function $\omega=\frac{u}{\sigma}+v\ln\sigma+w$ with $\sigma(x,y)$ denoting the square of the length of the geodesic connecting $x$ and $y$ and with geometrically determined smooth functions $u$ and $v$. A more explicit construction of Wick products was provided in a recent review \cite{FR15}. In the present case the only difference lies in the fact that elements of $\BVcal_{\loc}(\Mcal)$ are typically formal power series in $\la$, with coefficients that are local polynomials of arbitrary. As an example, we consider the Wick ordered scalar curvature on a background $g_0$.
\begin{multline*}
{\TT_1}_{\Mcal}(\Phi_{\vr{f}})=\int_M R[g_0]\vr{f}(X_{g_0})d\mu_{g_0}+\\
+\la\,\al^{-1}_{H-w}\left( \int_M \vr{f}(X_{g_0})\left.\frac{\delta}{\delta g}(Rd\mu)\right|_{g_0}(h)+\int_M R[g_0]\partial_\mu \vr{f}(x_g) \left.\frac{\delta X^\mu_g}{\delta g}\right|_{g_0}(h)\right)+\Ocal(\hbar^2)\, .
\end{multline*}
For the simplicity of notation we will drop the subscript $\Mcal$ if we keep the background $\Mcal$ fixed and use the notation ${\TT_1}$ instead of ${\TT_1}_{\Mcal}$ for the Wick ordering operator.
\subsection{Interacting theory}
Following \cite{FR3}, we introduce the interaction by means of renormalized time-ordered products. Let $\Delta_D\doteq \frac{1}{2}(\Delta^R+\Delta^A)$ denote the Dirac propagator. By $\fA_{\loc}(\Mcal)$ denote the space $\TT_1(\BVcal_{\loc}(\Mcal)[[\hbar]])$ of Wick ordered local functionals and we define operators $\Tcal_{n}:\fA_{\loc}(\Mcal)^{\otimes n}\rightarrow \fA(\Mcal)$, $n>1$ by means of 
\[
\Tcal_{n}(F_1,\ldots,F_n)=\al_{H-w}^{\minus} (F_1)\T\ldots\T \al_{H-w}^{\minus} (F_n)\,,
\]
for $F_i\in \fA_{\loc}(\Mcal)$ with disjoint supports\footnote{Note that $F_i$, $i=1,\dots,n$ are of the form $\Phi^i_{\vr{f}_i}$  for some locally covariant quantum field $\Phi$. By pairwise disjoint supports we therefore mean that the supports of $\vr{f}_i$ are pairwise disjoint.},  where 
\[
F\T G\doteq m\circ \exp({i\hbar \Gamma'_{\Delta_D}})(F\otimes G),
\]
and we set $\TT_{0}=1$, $\TT_1=\al^{-1}_{H+w}$.
 Maps $\Tcal_{n}$ have to be extended to functionals with coinciding supports and are required to satisfy the standard conditions given in \cite{BDF,H}. In particular, we require graded symmetry, unitarity, scaling properties, $\supp\TT_n(F_1,\dots,F_n)\subset\bigcup\supp F_i$ and causal factorization property:
if the supports of $F_1\ldots F_i$ are later than the supports of $F_{i+1},\ldots F_n$, then we have
\be\label{CausFact}
\TT_{n}(F_1\otimes \dots \otimes F_n)=
\TT_{i}(F_1\otimes \dots \otimes F_i) \star
\TT_{n-i}(F_{i+1} \otimes \dots \otimes F_n) \, .
\ee
Maps satisfying the conditions  above  are constructed inductively, and $\TT_n$ is uniquely fixed by the lower order maps $\TT_k$, $k<n$, up to the addition of an $n$-linear  map
\be
\Zcal_n:\fA_\loc(\Mcal)^n\to\al_{H+w}^{-1}(\fA_\loc(\Mcal))=:\fA_\loc(\Mcal)\,,
\ee
which describes possible finite renormalizations.
In \cite{FR3} it was shown that the renormalized time ordered product can be extended to an associative, commutative binary product defined on the domain $\Dcal_{\TT}(\Mcal)\doteq\TT(\BVcal(\Mcal))$, where $\TT\doteq\oplus_n\TT_n\circ m^{-1}$.
Here $m^{-1}:\BVcal(\Mcal)\to S^\bullet\BVcal^{(0)}_\loc(\Mcal)$ is the inverse of the multiplication, as defined in \cite{FR3,Rej11b}. The only difference is that now we consider functionals that are formal power series in $\la$. $\Dcal_{\TT}(\Mcal)$ contains in particular $\fA_\loc(\Mcal)$ and is invariant under the renormalization group action. Renormalized time ordered products are defined by
\be
F\T G\doteq\TT(\TT^{\minus}F\cdot\TT^{\minus}G)\,,
\ee
and we use the notation $\normOrd{F}\doteq \TT(F)$.
Time ordered products on different spacetimes have to be defined in a covariant way. To show that this can be done, one uses a straightforward generalization of the result of \cite{H} on the existence of covariant time-ordered products for Yang-Mills theories.

Using covariant time-ordered products we can now introduce the interaction.
As indicated in section \ref{free:theory}, we split the action into $L^\ex=L_0+L_I$, where $L_I$ is the interaction term. Let $\vr{f}\doteq(\vr{f}_0,\vr{f}_1)$ be a tuple of test functions chosen in such a way that $\vr{f}_i(X_{g_0})$, $i=0,1$ are compactly supported. We require that $\vr {f}_0\equiv 1$ on $\supp\vr{f}_1$ (compare with the condition preceding \eqref{pairing}) and we have a pairing $L^\ex_{\vr{f}}=L_{0\vr{f}_0}+L_{I\vr{f}_1}$.

The formal S-matrix $\Scal$ is a map from $\fA_{\loc}(\Mcal)$ to $\fA(\Mcal)$ defined as the time-ordered exponential. In particular, we have
\be\label{Smatrix}
\Scal(\normOrd{L_{I\vr{f}}})= e_{\sst{\TT}}^{i\TT L_{I\vr{f}}/\hbar}\,.
\ee
Now we want to construct a local net of $*$-algebras corresponding to the interacting theory on a fixed spacetime $\Mcal$. This is done along the lines of \cite{BDF}, by means of relative S-matrices. For $V,F\in\fA_{\loc}(\Mcal)$ the relative S-matrix is defined by the Bogoliubov formula
\be\label{SV}
\Scal_V(F)\doteq \Scal(V)^{-1}\star\Scal(V+F)\,.
\ee
The infinitesimal version of the above formula allows to define an interacting field corresponding to an observable $F$ under the influence of the interaction $V$:
\be\label{RV}
R_V(F)\doteq -i\hbar \frac{d}{ds}\left.\Scal_V(sF)\right|_{s=0}=\left(e_{\sst{\TT}}^{i\TT V/\hbar}\right)^{\star\minus}\star\left(e_{\sst{\TT}}^{i\TT V/\hbar}\T \TT F\right)\,.
\ee
Unfortunately, we cannot insert directly $\normOrd{L_{I\vr{f}}}$ as $V$, since the resulting interacting fields would in general depend on the choice of the cutoff function $\vr{f}$. One way to do it would be to take the limit $f\rightarrow 1$ directly, in some appropriate topology. This, however, is typically not well under control. Instead we construct the so called ``algebraic adiabatic limit''. 

Let $\Ocal$ be a relatively compact open subregion of the spacetime $\Mcal$. From the support properties of the retarded M{\o}ller operator follows that for $F\in\fA_{\loc}(\Ocal)$, the S-matrix $\Scal_{L_{I\vr{f'}}}(F)$ depends only on the behavior of $f'\doteq \vr{f'}\circ X_{g_0}$ within $J_-(\Ocal)$. Moreover, the dependence on $f'$ in that part of the past which is outside of $J_+(\Ocal$) is described by a unitary transformation which is independent of $F$. Concretely, if $f''=\vr{f}''\circ X_{g_0}$ coincides with $f'$ on a neighborhood of $J^{\diamond} (\Ocal):= J_+(\Ocal) \cap J_-(\Ocal)$, then there exists a unitary $U(f'', f') \in \fA[[\hbar]]$  (formal power series in $\hbar$, $\la$ and possibly $\mu$) such that
\[
\Scal_{L_{I\vr{f''}}}(F)= U(f'', f')\Scal_{L_{I\vr{f}'}}(F)U(f'', f')^{-1}\,,
\]
for all $F\in\fA_{\loc}(\Ocal)$. Hence the algebra generated by the elements of the form  $\Scal_{L_{I\vr{f}'}}(F)$  is, up to isomorphy, uniquely determined by the restriction of $f'$ to the causal completion $J^{\diamond} (\Ocal)$. This defines an abstract algebra $\fA_{L_I[f']}(\Ocal)$, where  $[f']\equiv[f']_\Ocal$ denotes the class of all test functions which coincide with $f'$ on a neighborhood of $J^{\diamond}(\Ocal)$. In fact, $f'$ can be chosen as a smooth function without the restriction on the support. The algebra $\fA_{L_I[f']}(\Ocal)$, is generated by maps

\[
R_{L_{I}[f']}(F):[f']_{\Ocal}\rightarrow \fA(\Mcal),\quad f'\mapsto R_{L_{I\vr{f}'}}(F)=i\hbar\frac{d}{d\la}\Scal_{L_{I\vr{f}'}}(\la F)\Big|_{\la=0}\,.
\]
Now if $\Ocal_1\subset \Ocal_2$, we can then define a map $\fA_{L_I[f']}(\Ocal_1)$ to $\fA_{L_I[f']}(\Ocal_2)$ by taking the restriction of maps $R_{L_{I}[f']_{\Ocal_1}}(F)$ to $[f']_{\Ocal_2}$. For $f'=1$ we denote $\fA_{L_I[1]}(\Ocal)\equiv \fA_{S_I}(\Ocal)$ and analogously $R_{L_{I}[1]}(F)\equiv R_{S_{I}}(F)$ for $F\in\fA_\loc(\Ocal)$. We can now construct the inductive limit $\fA_{S_I}(\Mcal)$ of the net of local algebras  $(\fA_{S_I}(\Ocal))_{\Ocal\subset\Mcal}$. We call this the \textit{algebraic adiabatic limit}. 

Note that for $V\in \BV_{\reg}(\Mcal)$ we can define on  $\BV_\reg(\Mcal)$ a product $\star_V$ as
 \be\label{starV}
 F\star_V G\doteq R_V^{-1}(R_V(F)\star R_V(G))\, .
 \ee
This doesn't work for local arguments, as $R_V^{-1}$ would not be well defined. Instead, we can define $\star_{S_I}$ formally, by setting
 \be\label{starSI}
R_{S_I}( F\star_{S_I} G)\doteq R_{S_I}(F)\star R_{S_I}(G)\, .
 \ee

Let us now come back to quantization of structures appearing in the BV formalism. Following the approach proposed in \cite{FR3},
we define the renormalized time-ordered antibracket on $\TT(\BVcal(\Mcal))$ by
\[
\{X,Y\}_{\sst{\TT}}=\TT\{\TT^{-1}X,\TT^{-1}Y\}\ .
\]
We can also write it as:
\be\label{antibracketTR}
\{X,Y\}_{\TT}=\sum_\al\int\!\left(\!\frac{\delta^r X}{\delta\ph^\al}\T\frac{\delta^l Y}{\delta\ph_\al^\dgr}-(-1)^{|\ph_\al^\dgr|}\frac{\delta^r X}{\delta\ph_\al^\dgr}\T\frac{\delta^l Y}{\delta\ph^\al}\!\right)d\mu\, .
\ee
The above formula  has to be understood as:
\be\label{antibracketTR2}
\{F,G\}_{\TT}\doteq \TT\Big(D^*\Big(\TT ^{-1}\frac{\delta F}{\delta\ph}\otimes \TT ^{-1}\frac{\delta G}{\delta\ph^\ddagger}\Big)\Big)\,,
\ee
where $D^*$ denotes the pullback by the diagonal map and $\big(\TT ^{-1}\frac{\delta F}{\delta\ph}\big)(\ph)$ is a compactly supported distribution (i.e. an element of $\E'(\Mcal)$) defined by
\[
\left<\big(\TT ^{-1}\frac{\delta F}{\delta\ph}\big)(\ph),f\right>\doteq\Big(\TT ^{-1}\Big<\frac{\delta F}{\delta\ph},f\Big>\Big)(\ph)=\Big<\frac{\delta}{\delta\ph}\TT ^{-1}F,f\Big>(\ph)\,,\qquad f\in\E(\Mcal)\,.
\]
In the second step we used the field independence of time ordered products. Since $F\in\TT(\BV(\Mcal))$, the distribution $\big(\TT ^{-1}\frac{\delta F}{\delta\ph}\big)(\ph)$ defined by the above equation is actually a smooth function and the pullback in \eqref{antibracketTR2} is well defined.
Similarly, we define the antibracket with the $\star$-product:
\be\label{antibracketstar}
\{X,Y\}_{\star}=\sum_\al\int\!\left(\!\frac{\delta^r X}{\delta\ph^\al}\star\frac{\delta^l Y}{\delta\ph_\al^\dgr}-(-1)^{|\ph_\al^\dgr|}\frac{\delta^r X}{\delta\ph_\al^\dgr}\star\frac{\delta^l Y}{\delta\ph^\al}\!\right)d\mu\,,
\ee
whenever it exists. Clearly, it is well defined if one of the arguments is regular or equal to $S_0$. The antibracket $\{.,S_0\}_\star$ with the free action defines a $\star$-derivation and, similarly, $\{.,S_0\}_{\sst{\TT}}$ is a $\T$-derivation. A relation between these two is provided by the Master Ward Identity \cite{BreDue,H}:
\be\label{MWI}
\{e_{\sst{\TT}}^{iV/\hbar}\T X,S_0\}_\star=\{e_{\sst{\TT}}^{iV/\hbar}\T X,S_0\}_{\sst{\TT}}+e_{\sst{\TT}}^{iV/\hbar}\TR(i\hbar\Lap_V(X)+\{X,V\}_{\sst{\TT}})\,.
\ee
Now we can use the BV formalism to discuss the gauge invariance in the quantum theory. In the framework of \cite{FR3}, the S-matrix is independent of the gauge fixing-fermion if the quantum master equation (QME) is fulfilled on the level of natural transformations. In terms of the relational observables we use in the present work, this condition means that at each order in $\la$ and $\hbar$,
\be\label{QME2}
\supp\left(e_{\sst{\TT}}^{-i\TT L_{I\vr{f}_1}/\hbar}\T\left\{e_{\sst{\TT}}^{i \TT L_{I\vr{f}_1}/\hbar},L_{0\vr{f}_0}\right\}_{\star}\right)\subset \supp(df_1)\,,
\ee
where $f_1\doteq\vr{f}_1\circ X_{g_0}$.
Using the Master Ward Identity \cite{BreDue,H} and our choice of $\vr{f}_1$, $\vr{f}_0$, we can rewrite the above condition as:
\be\label{qme}
\supp\left(\{L^\ex_{\vr{f}},L^\ex_{\vr{f}}\}+\Lap(L_{I\vr{f}})\right)\subset \supp(df_1)\,,
\ee
where $\Delta(L_{I\vr{f}})$ is the anomaly term, which in the formalism of \cite{FR3}, is interpreted as the renormalized version of the BV Laplacian. 
The condition \eqref{qme} is called \textit{the quantum master equation}. If we redefine time-ordered products in such a way that the anomaly is equal to 0, the above condition is fulfilled. To show that such a redefinition of time-ordered products is possible, one uses a cohomological argument similar to that of \cite{H,FR3}, which reduces the problem of removing the anomaly term to the problem of analyzing the cohomology of $s$ modulo $d$ on local forms (forms constructed locally from the fields of the theory). For the case of gravity  in the metric formulation, the relevant cohomology (i.e. $H^1(s|d)$ on top forms) was computed in \cite{BBH95} (see also earlier work \cite{BTM}, without antifields). In 4 dimensions for pure gravity this cohomology is trivial, so the anomaly can be removed, i.e. one can redefine the time-order products in such a way that \eqref{qme} holds for the new definition of $\Tcal$.

Let us now define the quantum BV operator $\hat{s}$, as a map on $\BVcal(\Mcal)$ given by
\be\label{QBV}
\hat{s}(X)=e_{\sst{\TT}}^{-i\TT L_{I\vr{f}}/\hbar}\T\Big(\left\{e_{\sst{\TT}}^{i\TT L_{I\vr{f}} /\hbar}\T  \TT X,L_{0\vr{f}})\right\}_{\star}\Big)-\{L^\ex_{\vr{f}},L^\ex_{\vr{f}}\}_{\sst\TT}\T \TT X\,,
\ee
where the second term is a correction for the fact that  $\{L^\ex_{\vr{f}},L^\ex_{\vr{f}}\}_{\sst\TT}$ vanishes only for $\vr{f}\rightarrow 1$. The nilpotency of $\hat{s}$ is easily checked by direct computation, with the use of the Jacobi identity for the antibracket and the fact that $\{L^\ex_{\vr{f}},L^\ex_{\vr{f}}\}_{\sst\TT}$ is odd.
From the {\mwi} follows that $\hat{s}$ can be rewritten as
\[
\hat{s}(X)=\{X,S^\ex\}+\Delta_{S_{I}}(X)\,,
\]
so it is local and doesn't depend on the choice of $\vr{f}$. As in \cite{FR3} we have an intertwining property
\be\label{intertwining:s:r}
\{.,S_0\}_\star \circ R_{L_{I\vr{f}}}=R_{L_{I\vr{f}}}\circ\hat{s}+(\textrm{terms that vanish for }d\vr{f}=0)\,,
\ee
hence we can formally state that
\[
\hat{s}=\left.R_{L_{I\vr{f}}}^{-1}\circ\{.,S_0\}_\star \circ R_{L_{I\vr{f}}}\right|_{d\vr{f}=0}\,.
\]
Note that $\hat{s}$ doesn't depend on  the choice of $\vr{f}$ and the intertwining property above suggests that $\hat{s}$ should (at least formally) be a derivation with respect to $\star_{S_I}$. To make this statement precise, we can use the fact that $\hat{s}$ is locally implemented by the BRST charge $Q$ \cite{Rej13}. It is defined as the Noether charge corresponding to the BRST transformation. A concrete formula is provided in Appendix \ref{free:charge}. Let us assume that $\Mcal$ has a compact Cauchy surface. Using the result of \cite{Rej13} we can conclude that
\be\label{chargecom}
R_{L_{I\vr{f}}}(\hat{s}\Phi_{\vr{f}'})=\frac{i}{\hbar}[R_{L_{I\vr{f}}}(\Phi_{\vr{f}'}),R_{L_{I\vr{f}}}(Q)]_{\star}
\ee
holds on-shell for $\Phi_{\vr{f}'}\in\widetilde{\BVcal}(\Ocal)$, where $f'\doteq \vr{f}'\circ X_{g_0}$ is supported in $\Ocal$ and $f\doteq \vr{f}\circ X_{g_0}$ is identically 1 on $\Ocal$. Formally, this can be written as
\[
\hat{s}\Phi_{\vr{f}'}=[\Phi_{\vr{f}'},Q]_{\star_{L_{I\vr{f}}}}\,.
\]
As we are interested in constructing only the local algebras associated to bounded regions $\Ocal\subset\Mcal$, we can always embed such a region in a spacetime with compact Cauchy surfaces. Since the $\star_{L_{I\vr{f}}}$--commutator is local, it doesn't depend on the behavior of $Q$ in the region spacelike to the support of $f'$, so the formula \eqref{chargecom} holds also for spacetimes with non-compact Cauchy surfaces, although $Q$ alone is not well defined (see the remarks in \cite{H} at the end of section 4.1.1).

%

We can now define the space of gauge invariant fields as the 0th cohomology of $(\hat{s},\BVcal(\Mcal))$. This concludes the construction of the algebra of diffeomorphism invariant quantum fields for general relativity.
\section{Background independence}\label{background}
In the previous section we constructed the algebra of interacting observables of effective quantum gravity, by choosing a background and splitting the action into a free and interacting part. Now we prove that the result is independent of that split. In \cite{BF1} it was proposed that a condition of background independence can be formulated by means of the relative Cauchy evolution. Let us fix a spacetime $\Mcal_1=(M,g_1)\in\obj(\Loc)$ and choose $\Sigma_-$ and $\Sigma_+$, two Cauchy surfaces in  $\Mcal_1$, such that $\Sigma_+$ is in the future of $\Sigma_-$. Consider another globally hyperbolic metric $g_2$ on $M$, such that $k\doteq g_2-g_1$ is compactly supported and its support $K$ lies between  $\Sigma_-$ and $\Sigma_+$.  Let us take $\Ncal_{\pm}\in\obj(\Loc)$ that embed into $\Mcal_1$, $\Mcal_2$, via $\chi_{1\pm}$, $\chi_{2\pm}$ and $\chi_{i\pm}(\Ncal_{\pm})$, $i=1,2$ are causally convex 
neighborhoods of  $\Sigma_\pm$ in   $\Mcal_i$. 
We can then use the time-slice axiom to define isomorphisms $\al_{\chi_{i\pm}}\doteq\fA\chi_{i\pm}$ and the free relative Cauchy evolution is an automorphism of $\fA(\Mcal_1)$ given by $\beta_{0g}=\al^{\phantom{-1}}_{0\chi_{1-}}\circ\al^{-1}_{0\chi_{2-}}\circ\al^{\phantom{-1}}_{0\chi_{2+}}\circ\al_{0\chi_{1+}}^{-1}$. It was shown in \cite{BFV} that the functional derivative of $\beta$ with respect to $g$ is the commutator with the free stress-energy tensor. Let us recall briefly that argument, using a different formulation. We can apply $\beta$ to the S-matrix, which works as the generating function for free fields, and calculate the functional derivative using an explicit formula for relative Cauchy evolution. To this end we use the \textit{perturbative agreement} condition introduced by Hollands and Wald in \cite{HW5}. Recently a more general result in this direction was proven in \cite{DHP}. Following these ideas, we use a map $\tau^{\textrm{ret}}:\fA(\Mcal_2)\rightarrow\fA(\Mcal_1)$, such that $\tau^{\textrm{ret}}$ maps $\Phi_{\Mcal_2}(f)$ to 
$\Phi_{\Mcal_1}(f)$ (modulo the image of $\delta_0$), $f\equiv \vr{f}\circ X_{g_0}$, if the support of $f$ lies outside the causal future of $K$. Physically it means that free algebras $\fA(\Mcal_1)$ and $\fA(\Mcal_2)$ are identified in the past of $K$. Analogously, one defines a map  $\tau^{\textrm{adv}}$, which identifies the free algebras in the future. The free relative Cauchy evolution is then given by
\be\label{beta:adv:ret}
\beta_{0g}\doteq\tau^{\textrm{ret}}_{g_1g_2}\circ(\tau^{\textrm{adv}}_{g_1g_2})^{-1}\,,
\ee
As we choose to work off-shell, we define $\tau^{\textrm{ret}}$ as the classical retarded M{\o}ller operator constructed in \cite{DF02}. This definition can be understood as an off-shell extension of the definition given in \cite{HW5}.
The perturbative agreement is a condition that, on shell,
\be\label{perturbative:agreement}
\tau_{g_1g_2}^{\textrm{ret}}\circ\Scal_{2}=\Scal_{S_{0\sst \Mcal_2}-S_{0\sst \Mcal_1}}\qquad\textrm{holds.}
\ee
Here $\Scal_{S_{0\sst \Mcal_1}-S_{0\sst \Mcal_2}}$ denotes the relative S-matrix constructed with the interaction $S_{0\sst \Mcal_1}-S_{0\sst \Mcal_2}$ and the background metric $g_1$, while $\Scal_2$ is the S-matrix constructed on $\Mcal_2$ with the $\Tcal_{\Mcal_2}$ product. More explicitly, we have
\be\label{pa:ret}
\tau_{g_1g_2}^{\textrm{ret}}\left(e_{\sst{\TT_{\sst\Mcal_2}}}^{i\Phi_{\sst\Mcal_2\vr{f}'}/\hbar}\right)\os\left(e_{\sst{\TT_{\sst\Mcal_1}}}^{i(L_{0\sst \Mcal_2}-L_{0\sst \Mcal_1})_{\vr{f}}/\hbar}\right)^{-1}\star_{g_1}\left(e_{\sst{\TT_{\sst\Mcal_1}}}^{i(L_{0\sst \Mcal_2}-L_{0\sst \Mcal_1})_{\vr{f}}/\hbar+i\Phi_{\sst\Mcal_2{\vr{f}'}}/\hbar}\right)\,,
\ee
where $\os$ means ``holds on-shell with respect to free equations of motion'' (i.e. modulo the image of 
$\delta_0$) and, using the notation introduced in the previous section,  $(L_{0\sst \Mcal_1})_{\vr{f}}=(L_{0\sst \Mcal_1})_{\vr{f}_0}$, where $\vr{f}=(\vr{f}_0,\vr{f}_1)$ is a tuple of test functions such that $\vr{f}_0\equiv 1$ on $\supp \vr{f}_1$. We also choose $\vr{f}$ to be identically $(1,1)$ on $\supp \vr{f}'$.

The perturbative agreement condition for $\tau_{g_1g_2}^{\textrm{adv}}$ is analogous to \eqref{pa:ret} and reads:
\be\label{pa:adv}
\tau_{g_1g_2}^{\textrm{adv}}\left(e_{\sst{\TT_{\sst\Mcal_2}}}^{i\Phi_{\sst\Mcal_2\vr{f}'}/\hbar}\right)\os\left(e_{\sst{\TT_{\sst\Mcal_1}}}^{i(L_{0\sst \Mcal_2}-L_{0\sst \Mcal_1})_{\vr{f}}/\hbar+i\Phi_{\sst\Mcal_2{\vr{f}'}}/\hbar}\right)\star_{g_1}\left(e_{\sst{\TT_{\sst\Mcal_1}}}^{i(L_{0\sst \Mcal_2}-L_{0\sst \Mcal_1})_{\vr{f}}/\hbar}\right)^{-1}\,,
\ee
Conditions \eqref{pa:ret} and \eqref{pa:adv} were proven in \cite{HW5} for the case of the free scalar field, but the same argument can be used also for pure gravity. 

To fulfill the perturbative agreement condition,  one fixes the time-ordered product $\TT_{\Mcal_1}$ and shows that there exists a definition of $\TT_{\Mcal_2}$ on the background $\Mcal_2$ compatible with other axioms, such that also \eqref{perturbative:agreement} can be fulfilled. In particular, the quantum master equation holds automatically for $\TT_{\Mcal_2}$ if it holds for $\TT_{\Mcal_1}$. To prove this, we use the off-shell definition of $\tau_{g_1g_2}^{\textrm{ret}}$, given in \cite{DF02}, and from \eqref{perturbative:agreement} it follows that $\tau_{g_1g_2}^{\textrm{ret}}\circ\Scal_{2}(\Phi_{\Mcal_2\vr{f}'})=\Scal_{S_{0\sst \Mcal_2}-S_{0\sst \Mcal_1}}(\Phi_{\Mcal_2\vr{f}'})+I$, where $I$ belongs to the image of $\{.,S_{0\Mcal_1}\}_{\star_{g_1}}$. Let 
\[
V_i\doteq \TT_{\Mcal_i}(L_{\Mcal_i}-L_{0\Mcal_i})_{\vr{f}}\,.
\]
 Since $\tau_{g_1g_2}^{\textrm{ret}}$ is an algebra morphism and it maps $\frac{\delta S_{0\Mcal_2}}{\delta \ph(x)}$ to  $\frac{\delta S_{0\Mcal_1}}{\delta \ph(x)}$, it follows that
\begin{multline*}
\tau_{g_1g_2}^{\textrm{ret}}\left(\left\{e^{iV_2/\hbar}_{\TT_{\Mcal_2}},S_{0\Mcal_2}\right\}_{\star_{g_2}}\right)=\left\{\tau_{g_1g_2}^{\textrm{ret}}\left(e^{iV_2/\hbar}_{\TT_{\Mcal_2}}\right),S_{0\Mcal_1}\right\}_{\star_{g_1}}=\\=
\left\{\left(e_{\sst{\TT_{\sst\Mcal_1}}}^{i(L_{0\sst \Mcal_2}-L_{0\sst \Mcal_1})_{\vr{f}}/\hbar}\right)^{-1}\star_{\sst g_1}\left(e_{\sst{\TT_{\sst\Mcal_1}}}^{i((L_{0\sst \Mcal_2}-L_{0\sst \Mcal_1})_{\vr{f}}+V_2)/\hbar}\right),S_{0\Mcal_1}\right\}_{\star_{g_1}}
\end{multline*}
Now we use the fact that $(L_{0\sst \Mcal_2}-L_{0\sst \Mcal_1})_{\vr{f}}$ doesn't depend on antifields and that $(L_{0\sst \Mcal_2}-L_{0\sst \Mcal_1})_{\vr{f}}+V_2=V_1$. This yields
\[
\tau_{g_1g_2}^{\textrm{ret}}\left(\left\{e^{iV_2/\hbar}_{\TT_{\Mcal_2}},S_{0\Mcal_2}\right\}_{\star_{g_2}}\right)=\left(e_{\sst{\TT_{\sst\Mcal_1}}}^{i(L_{0\sst \Mcal_2}-L_{0\sst \Mcal_1})_{\vr{f}}/\hbar}\right)^{-1}\star_{\sst g_1}\left\{e^{iV_1/\hbar}_{\TT_{\Mcal_1}},S_{0\Mcal_1}\right\}_{\star_{g_1}}=0\,,
\]
so the {\qme} holds for $\TT_{\Mcal_2}$.

Let us go back to the relative Cauchy evolution. 
The functional derivative of $\beta_{0g}$ with respect to $k\doteq g_2-g_1$ can now be easily calculated, yielding
\begin{align*}
&\frac{\delta}{\delta k_{\mu\nu}} \beta_{0g}\left(e_{\sst{\TT_{\sst\Mcal_1}}}^{i\Phi_{{\sst\Mcal_1}\vr{f}'}/\hbar}\right)\Big|_{g_1}\\
&\qquad\qquad\os\frac{i}{\hbar}\left(\left(\frac{\delta (L_{0\sst \Mcal_2})_{\vr{f}}}{\delta k_{\mu\nu}}\Big|_{g_1}\right)\star_{g_1}e_{\sst{\TT_{\sst\Mcal_1}}}^{i\Phi_{{\sst\Mcal_1}\vr{f}'}/\hbar}-e_{\sst{\TT_{\sst\Mcal_1}}}^{i\Phi_{{\sst\Mcal_1}\vr{f}'}/\hbar}\star_{g_1}\left(\frac{\delta (L_{0\sst \Mcal_2})_{\vr{f}}}{\delta k_{\mu\nu}}\Big|_{g_1}\right)\right)\\
&\qquad\qquad=\frac{i}{\hbar}\left[T_{0\mu\nu},e_{\sst{\TT_{\sst\Mcal_1}}}^{i\Phi_{{\sst\Mcal_1}\vr{f}'}/\hbar}\right]_\star\,,
\end{align*}
where $T_{0\mu\nu}$ is the stress-energy tensor of the linearized theory.

Let us now discuss a corresponding construction in the interacting theory. It was conjectured in \cite{BF1} that, for the full interacting theory of quantum gravity, the relative Cauchy evolution should be trivial (equal to the identity map), hence the derivative with respect to $g$ should vanish. Using the  quantum M{\o}ller maps $R_{V_i}$, $A_{V_i}$, $i=1,2$, we can write the interacting relative Cauchy evolution as:
\[
\beta=R_{V_1}^{-1}\circ\tau_{g_1g_2}^{\textrm{ret}}\circ R_{V_2}\circ A_{V_2}^{-1}\circ (\tau_{g_1g_2}^{\textrm{adv}})^{-1}\circ A_{V_1}\,.
\]
We can now formulate the condition of background independence as:
\[
R_{V_1}^{-1}\circ\tau_{g_1g_2}^{\textrm{ret}}\circ R_{V_2}=A_{V_1}^{-1}\circ\tau_{g_1g_2}^{\textrm{adv}}\circ A_{V_2}\,.
\]
Note that we can avoid potential problems with domains of definition of $R_{V_1}^{-1}$ and $A_{V_1}^{-1}$, by rewriting the above condition as
\[
e_{\sst{\TT_{\sst\Mcal_1}}}^{iV_1/\hbar}\star_{g_1}(\tau_{g_1g_2}^{\textrm{ret}}\circ R_{V_2}(\Phi_{\Mcal_2\vr{f}'}))=(\tau_{g_1g_2}^{\textrm{adv}}\circ A_{V_2}(\Phi_{\Mcal_2\vr{f}'}))\star_{g_1}e_{\sst{\TT_{\sst\Mcal_1}}}^{iV_1/\hbar}\,.
\]
Using formulas for $\tau_{g_1g_2}^{\textrm{ret}}$ and $\tau_{g_1g_2}^{\textrm{adv}}$ and the fact that $(L_{0\Mcal_2})_{\vr{f}}+V_2=L^\ex_{\Mcal_2\vr{f}}$, we obtain:
\begin{multline*}
e_{\sst{\TT_{\sst\Mcal_1}}}^{iV_1/\hbar}\star_{g_1}\left(e_{\sst{\TT_{\sst\Mcal_1}}}^{i(L^\ex_{\Mcal_2}-L_{0\Mcal_1})_{\vr{f}}/\hbar}\right)^{-1}\star_{g_1}e_{\sst{\TT_{\sst\Mcal_1}}}^{i(L^\ex_{\Mcal_2}-L_{0\Mcal_1})_{\vr{f}}/\hbar+i\Phi_{\Mcal_2\vr{f}'}/\hbar}\os\\\os
e_{\sst{\TT_{\sst\Mcal_1}}}^{i(L^\ex_{\Mcal_2}-L_{0\Mcal_1})_{\vr{f}}/\hbar+i\Phi_{\Mcal_2\vr{f}'}/\hbar}\star_{g_1}\left(e_{\sst{\TT_{\sst\Mcal_1}}}^{i(L^\ex_{\Mcal_2}-L_{0\Mcal_1})_{\vr{f}}/\hbar}\right)^{-1}\star_{g_1}e_{\sst{\TT_{\sst\Mcal_1}}}^{iV_1/\hbar}
\end{multline*}
Differentiating with respect to ${k}_{\mu\nu}$ yields a condition
\[
[R_{V_1}(\Phi_{\Mcal_1\vr{f}'}),R_{V_1}(T(\eta))]_\star\os0\,,
\]
where 
\[
T(\eta)\doteq \left<{T_{\mu\nu}}_{\vr{f}},\eta^{\mu\nu}\right>=\left<\frac{\delta L^\ex_{\sst \Mcal_2\vr{f}}}{\delta {k}^{\mu\nu}}\Big|_{g_1},\eta^{\mu\nu}\right>
\]
 is the full stress-energy tensor smeared with a test section $\eta$ and we chose $f\equiv 1$ on $\supp \eta$. We can write the above condition in a more elegant way, using the formal notation with $\star_{V_1}$, namely
\[
[\Phi_{\Mcal_1\vr{f}'},T(\eta)]_{\star_{V_1}}\stackrel{\mathrm{o.s._{V_1}}}{=}0\,,
\]
where $\stackrel{\mathrm{o.s._{V_1}}}{=}$ means ``holds on-shell with respect to the equations of motion of the full interacting theory''.
To prove that the infinitesimal background independence is fulfilled, we have to show that $T(\eta)=0$ in the cohomology of $\hat{s}$.
This is easily done, as  
\[
T(\eta)=\left<\frac{\delta S^\ex_{\sst \Mcal_2}}{\delta {k}_{\mu\nu}}\Big|_{g_1},\eta^{\mu\nu}\right>=\left<\frac{\delta S^\ex_{\sst \Mcal_2}}{\delta {h}_{\mu\nu}}\Big|_{g_1},\eta^{\mu\nu}\right>=s\left<h^\ddagger,\eta\right>=\hat{s}\left<h^\ddagger,\eta\right>\,,
\]
where $h$ is the perturbation metric. The last equality follows from the fact that the anomaly can always be removed for linear functionals \cite{BreDue}. This concludes the argument, so the theory is perturbatively background independent.
 \section{States}\label{states}

 Finally we come to the discussion of states. We start with outlining the construction of a state for the full interacting theory for on-shell backgrounds (i.e. backgrounds for which the metric is a solution to Einstein's equations), given a state for the linearised theory. We will use the method proposed in \cite{DF99} which relies on the gauge invariance of the linearized theory under the free BV transformation $s_0$. We have already indicated that this requires the background metric $g_0$ to be a solution of the Einstein's equation, so throughout this subsection we assume that this is indeed the case. The construction we perform is only formal, since we don't control the convergence of interacting fields and we treat them as formal power series in $\hbar$ and $\lambda$. 
 
For a fixed spacetime $\Mcal=(M,g_0)$, we define the quantum algebra $\fA(\Mcal)$ of the free theory as in section \ref{free:theory}. Since we assumed in this subsection that $g_0$ is a solution of Einstein's equation, the free action $L_0$ contains only the term quadratic in $h$. 

Let us assume that we have a representation $\pi_0$ of $\fA(\Mcal)$ on an indefinite product space $\Kcal_0(\Mcal)$ and we denote $\Kcal(\Mcal)\doteq \Kcal_0(\Mcal)[[\hbar,\la]]$. The scalar product $\left<.,.\right>_{\Kcal(\Mcal)}$ on $\Kcal(\Mcal)$ is defined in terms of formal power series in $\hbar$ and $\la$. 
In order to distinguish a subspace of $\Kcal(\Mcal)$ that corresponds to physical states, we will apply the Kugo-Ojima formalism \cite{KuOji0,KuOji} that makes use of the interacting BRST charge $Q_{\textrm{int}}\equiv R_{S_I}(Q)$ to characterize the physical states in $\Kcal$. The nilpotency of $Q$ (as an operator on $\Kcal(\Mcal)$) can be shown by arguments analogous to \cite{H}, postulating appropriate Ward identities. It follows that the 0-th cohomology of $Q$ defines a space closed under the action of physical observables (i.e. under $H^0(\BVcal(\Mcal),\hat{s})$). To see that this is consistent, let us take $\Psi\in \ker(Q_{\textrm{int}})$ and $F\in\BVcal(\Mcal)$. Then
\[
R_{L_{I\vr{f}}}(\hat{s}F)\Psi= [R_{L_{I\vr{f}}}(Q),R_{L_{I\vr{f}}}(F)]_\star\Psi=R_{L_{I\vr{f}}}(Q) F\Psi
\]
holds, i.e. $R_{L_{I\vr{f}}}(\hat{s}F)\Psi\in\im(Q_{\textrm{int}})$, so it vanishes in the cohomology. Vectors belonging to 
 $\ker(Q_{\textrm{int}})$ are constructed perturbatively from the elements of $\ker(Q_0)\subset\Kcal_0(\Mcal)$ by the recursive method introduced in \cite{DF99}. The assumptions on $Q_0$ and $\Kcal_0(\Mcal)$ necessary for this method to work are the following:
 \begin{enumerate}
 \item $\left<\psi,\psi\right>_{\Kcal_0(\Mcal)}\geq 0$, $\forall \psi\in\Kcal_0(\Mcal)$,
 \item If $\psi\in\Kcal_0(\Mcal)$ satisfies $\left<\psi,\psi\right>_{\Kcal_0(\Mcal)}= 0$, then $\psi\in\Kcal_{00}(\Mcal)\equiv \ker Q_0$.
 \end{enumerate}
 It was shown in \cite{DF99} that under these assumptions $\left<.,.\right>_{\Kcal(\Mcal)}$ is positive definite on $\ker Q_{\textrm{int}}\subset \Kcal(\Mcal)$, so $H^0(Q,\Kcal(\Mcal))$ provides formally a Hilbert space representation of $H^0(\BVcal(\Mcal),\hat{s})$.
 
It remains to show that for a given on-shell background $\Mcal=(M,g_0)$ there exists a pre-Hilbert space representation $\Kcal_0(\Mcal)$ of  the quantum linearized theory satisfying the conditions above. This problem hasn't been solved yet in a full generality, but there has been a lot of progress made in the recent years, see for example \cite{Fewster:2012bj,BDM}. A technical problem which we have to face is that construction of Hadamard states is difficult in generic spacetimes. On the other hand, if a background $\Mcal$ has symmetries, it might happen that there is no sensible choice of curvature scalars $X^{\mu}_{g_0}$. Therefore, instead of looking at pure gravity, in concrete models it might be better to consider coupling to matter fields and make the coordinates $X^{\mu}$ dependent on these fields. A natural candidate is the Brown-Kucha\v{r} model \cite{BrK}, where the coordinates are fixed by four scalar ``dust fields''. The construction of the algebra of observables in such a model proceeds analogous to the one presented in this work. We plan to investigate such models in our future work and compare the results to the other approaches to quantum gravity \cite{GravQuant}.
\section{Conclusions and Outlook}
We showed in this paper how the conceptual problems of a theory of quantum
gravity can be solved, on the level of formal power series. The crucial new
ingredient was the concept of local covariance \cite{BFV} by which a theory is
formulated simultaneously on a large class of spacetimes. Based on this
concept, older ideas could be extended
and made rigorous. The construction uses the renormalized Batalin Vilkovisky
formalism as recently developed in \cite{FR3}.

In the spirit of algebraic quantum field theory \cite{Haag} we first
constructed the algebras of local observables. In a theory of gravity, this
is a subtle point, since on a first sight one might think that in view of
general covariance local observables do not exist. We approached this problem
in the following way. Locally covariant fields are, by definition,
simultaneously declared on all spacetimes. These objects then give rise to partial (relational) observables used by Rovelli \cite{Rovelli:2001bz},  Dittrich \cite{Dittrich:2005kc} and Thiemann \cite{Thiemann:2004wk}.  The algebra of observables is defined as being generated by such objects.

The states in the algebraic approach are linear functionals on
the algebra of observables interpreted as expectation values. In gauge theories the algebra of observables is obtained as the cohomology of the BRST differential on an extended algebra.
The usual construction first described by Kugo and Ojima \cite{KuOji78,KuOji0,KuOji} (for an earlier attempt see \cite{CurFer}) starts from a representation of the  extended algebra on some Krein space and an implementation of the BRST differential as the graded commutator with a nilpotent (of order 2) operator (the BRST charge). The cohomology of this operator is then a representation space for the algebra of observables. We followed this approach also here, assuming there exists a representation of the linearized theory, and constructed as in \cite{DF} the full interacting theory as a formal power series in $\hbar$ and $\lambda$.

In this paper we treated pure gravity. It is, however, to be expected that the procedure can be easily extended to include matter fields (scalar, Dirac, Majorana, gauge). It is less clear whether supergravity can be treated in an analogous way. Introducing matter fields will make it easier to construct the dynamical coordinates $X^\mu$, for example like in the Brown-Kucha\v{r} models \cite{BrK}.

On the basis of the formalism developed in this paper one should be able to perform reliable calculations for quantum corrections to classical gravity, under the assumption that these corrections are small and allow a perturbative treatment. There exist already some calculations of corrections, e.g. for the Newton potential \cite{Bjerrum-Bohr} with which these calculations could be compared. It would also be of great interest to adapt the renormalization approach of Reuter et al. (see, e.g., \cite{Reu98,Reu02}) to our framework. Further interesting problems are the validity of the semiclassical Einstein equation (for an older discussion see \cite{Wald}) and the possible noncommutativity of the physical spacetime
\cite{Doplicher-Pinamonti}.

Another possible direction of further study would be to reformulate everything in terms of frames instead of a coordinate systems. The advantage of that is the existence of global frames in a large class of spacetimes, where global coordinate systems do not exist.
\section*{Acknowledgements}
We would like to thank Dorothea Bahns, Roberto Conti and Jochen Zahn for enlightening discussions and comments.
In particular we thank Jochen Zahn for pointing out a gap in our argument for background independence. K. R. is also grateful to INdAM (Instituto Nazionale di Alta Mathematica ``Francesco Severi'') for supporting her research and to the Erwin Schr{\"o}dinger Institute in Vienna for hospitality. Both authors would like to thank one of the referees for suggesting the example of scalar fields coupled to gravity as an illustration of our concept of metric dependent coordinates, which we have followed in section 2.5.  
\appendix
\section{Aspects of classical relativity seen as a locally covariant field theory}\label{class}
In this appendix we discuss some details concerning the formulation of classical relativity in the framework of locally covariant quantum field theory. The first issue concerns the choice of a topology on the configuration space $\E(\Mcal)$. In section \ref{conf:space} we already indicated that a natural choice of such a topology is  $\tau_W$, given by open neighborhoods of the form $U_{g_0,V}=\{g_0+h,h\in V\textrm{ open in }\Gamma_c((T^*M)^{\otimes2})\}$, where $\Gamma_c((T^*M)^{\otimes2})$ is equipped with the standard inductive limit topology. In our case, $\tau_W$ coincides with the Whitney $\mathcal{C}^\infty$ topology, $WO^\infty$, hence the notation. After \cite{Lerner}, Whitney $\mathcal{C}^\infty$ topology  is the initial topology on $\Ci(M,(T^*M)^{\otimes2})$  induced by the graph topology on $\Ci(M,J^\infty(M,(T^*M)^{\otimes2})$ through maps $\Gamma((T^*M)^{\otimes2}) \ni h\mapsto j^\infty h$, where $J^\infty(M,(T^*M)^{\otimes2})$ is the jet space and $j^\infty h$ is the infinite jet of $h$. On the space of all Lorentzian metrics we have also another natural topology, namely  the interval topology $\tau_I$ introduced by Geroch \cite{Geroch}, which is given by intervals $\{g|g_1 \prec g \prec g_2\}$, where the partial order relation  $\prec$ is defined by \eqref{partial:order}, i.e.
 \[
 g' \prec g\  \textrm{if}\  g'(X,X) \geq 0\  \textrm{implies}\  g(X,X) > 0\,.
 \]
The configuration space $\E(\Mcal)$, defined in \eqref{config} is, by definition, an open subset of  $\Lor(M)$, with respect to $\tau_I$.  Moreover, if $g'\in \E(\Mcal)$, then we know that there exists $\la\in\RR$ such that $\la g-g'$ is positive definite, so we can find a neighborhood $V\subset \Gamma_c((T^*M)^{\otimes2})$ of $0$,  such that $g'+h\in \Lor(M)$ and $\la g-g'-h$ is also positive definite. It follows that $g'+h <\la g$ and $g'+V\subset \E(\Mcal)$. This shows that $\E(\Mcal)$ is open also with respect to $\tau_W$. More generally, it was shown in \cite{Lerner} that the $\mathcal{C}^0$ Whitney topology, $WO^0$, on $\Lor(M)$ conincides with the interval topology on the space of continuous Lorentz metrics. This result was than used in \cite{BFR} to show that the space of smooth, time-oriented and globally hyperbolic Lorentzian metrics on $M$ is an open subset of  $\Lor(M)$, with respect to $WO^\infty$.

Functionals on $\E(\Mcal)$ are required to be smooth in the sense of  calculus on locally convex vector spaces, but the relevant topology is the compact open topology $\tau_{CO}$ not the Whitney topology $\tau_W$. More precisely,  let $U$ be an open neighborhood of $h_0$ in the compact open topology $\tau_{CO}$. The derivative of $F$ at $h_0$ in the direction of $h_1\in \Gamma((T^*M)^{\otimes 2})$ is defined as
\be\label{de}
\left<F^{(1)}(h_0),h_1\right> \doteq \lim_{t\rightarrow 0}\frac{1}{t}\left(F(h_0 + th_1) - F(h_0)\right)
\ee
whenever the limit exists. The function $F$ is called differentiable at $h_0$ if $\left<F^{(1)}(h_0),h_1\right>$ exists for all $h_1 \in \E(\Mcal)$. It is called continuously differentiable if it is differentiable at all points of $\E(\Mcal)$ and
$dF:U\times \E(\Mcal)\rightarrow \RR, (h_0,h_1)\mapsto \left<F^{(1)}(h_0),h_1\right>$
is a continuous map. It is called a $\Ccal^1$-map if it is continuous and continuously differentiable. Higher derivatives are defined in a similar way. Note that the above definition means that $F$ is 
smooth, in the sense of calculus on locally convex vector spaces, as a map $U\rightarrow \RR$. It was shown in \cite[Remark 2.3.9]{BFR} that this fits also into the manifold structure on $\E(\Mcal)$ induced by $\tau_W$. To see this, note that a compactly supported functional $F$, defined on a $\tau_W$-open set $U_{g_0,V}$ can be extended to a functional $F\circ \iota_\chi$ defined on an $\tau_{CO}$-open neighborhood $\iota_\chi^{-1}(U_{g_0,V})$ by means of a continuous map $\iota_\chi: (\Gamma((T^*M)^{\otimes 2}),\tau_{CO})\rightarrow (\Gamma((T^*M)^{\otimes 2}),\tau_{W})$, defined by $\iota_\chi(g')\doteq g_0+(g'-g_0)\chi$. From the support properties of $F$ follows that $F\circ \iota_\chi$ is independent of $\chi$. 

In particular,  $F^{(1)}$ defines a kinematical vector field on $\E(\Mcal)$ in the sense of \cite{Michor}. Moreover, since $\E_c(\Mcal)$ is reflexive and has the approximation property, it follows (theorem 28.7 of \cite{Michor})  that kinematical vector fields are also operational i.e., they are derivations of the space of smooth functionals on $\E(\Mcal)$.
 
At the end of section \ref{BVcomplex} we have indicated that the space of multilocal functionals can be extended to a space $\overline{\BV(\Mcal)}$ which is closed under $\Pei{.}{.}_{\tilde{g}}$. Here we give a possible choice for this space. We define $\overline{\BV(\Mcal)}$ to be a subspace of $\BV_{\mc}(\Mcal)$ (defined in section \ref{free:theory}) consisting of functionals $F$, such that the first derivative $F^{(1)}(\ph)$ is a smooth section for all $\ph\in\E(\Mcal)$ and 
$\ph\mapsto F^{(1)}(\ph)$ is smooth as a map $\E(\Mcal)\rightarrow \E(\Mcal)$, where $\E(\Mcal)$ is equipped with its standard Fr{\'e}chet topology. Since the lightcone of $\tilde{g}$ is contained in the interior of the lightcone of $g$, the WF set condition \eqref{mlsc} guarantees that  $\Pei{.}{.}_{\tilde{g}}$ is well defined on  $\overline{\BV(\Mcal)}$. Using arguments similar to \cite{BFR} we can prove the following proposition:
\begin{prop}
The space $\overline{\BV(\Mcal)}$ together with $\Pei{.}{.}_{\tilde{g}}$ is a Poisson algebra. 
\end{prop}
\begin{proof}
First we have to show that $\overline{\BV(\Mcal)}$ is closed under $\Pei{.}{.}_{\tilde{g}}$. It was already shown in \cite{BFR} that $\BV_{\mc}(\Mcal)$ is closed under the Peierls bracket. It remains to show that the additional condition we imposed on the first derivative is also preserved under $\Pei{.}{.}_{\tilde{g}}$. Consider 
\begin{align}\label{Pei:F:G:1}
(\Pei{F}{G}_{\tilde{g}})^{(1)}(\ph)&=\left<F^{(2)}(\ph),\Delta G^{(1)}(\ph)\right>-\left<\Delta F^{(1)}(\ph), G^{(2)}(\ph)\right>\nonumber
\\&\quad-\left<\Delta^AF^{(1)}(\ph), S'''(\ph) \Delta^RG^{(1)}(\ph)\right>
\\&\quad+\left<\Delta^RF^{(1)}(\ph), S'''(\ph) \Delta^AG^{(1)}(\ph)\right>\nonumber\,,
\end{align}
where $ S'''(\ph)$ denotes the third derivative of the action. The last two terms in the above formula are smooth sections, since
the wavefront set of $S'''(\ph)$ is orthogonal to $T\textrm{Diag}^3(M)$ and $\Delta^{R/A}F^{(1)}(\ph)$, $\Delta^{R/A}G^{(1)}(\ph)$ are smooth. The first term of \eqref{Pei:F:G:1} can be written as $\tfrac{d}{dt}F^{(1)}(\ph+th)\Big|_{t=0}$, where $h=\Delta G^{(1)}(\ph)$ is smooth. By assumption, $\ph\mapsto F^{(1)}(\ph)$ is smooth, so the above derivative exists as a smooth section in $\E(\Mcal)$. The same argument can be applied to the second term in \eqref{Pei:F:G:1}, so we can conclude that $\Pei{F}{G}_{\tilde{g}})^{(1)}(\ph)$ is a smooth section. From a similar reasoning follows also that $\ph\mapsto (\Pei{F}{G}_{\tilde{g}})^{(1)}(\ph)$ is a smooth map.

The antisymmetry of $\Pei{.}{.}_{\tilde{g}}$ is clear, so it remains to prove the Jacobi identity. In \cite{Jac,BFR} it was shown that this identity follows from the symmetry of the third derivative of the action, as long as products of the form $\Delta^{R/A}F^{(1)}(\ph)$ are well defined. With our definition of $\overline{\BV(\Mcal)}$ this is of course true, since $F^{(1)}(\ph)$ is required to be a smooth section.
\end{proof}
 \section{BRST charge}\label{free:charge}
 In this section we construct the BRST charge that generates the gauge-fixed BRST transformation $s$. It is convenient to pass from the original Einstein-Hilbert Lagrangian to an equivalent one given by:
\[
L'_{(M,g_0)}(f)(h)=\int\limits_M\textrm{dvol}_{(M,g)}g^{\mu\nu}\left({\Gamma}^\la_{\mu\rho}{\Gamma}^\rho_{\nu\lambda}-{\Gamma}^\rho_{\mu\nu}{\Gamma}^\la_{\rho\lambda}\right)\,.
\]
It differs from the Einstein-Hilbert Lagrangian by a term $\int\limits_M f\nabla_\mu\mathcal{D}^\mu$, where
\[
\mathcal{D}^\mu=\sqrt{-g}(g^{\rho\sigma}{\Gamma}^\mu_{\rho\sigma}-g^{\mu\nu}{\Gamma}^\la_{\nu\la})
\]
and $\Gamma$'s are the Christoffel symbols. 
Let $L$ be the gauge-fixed Lagrangian, where the Einstein-Hilbert term is replaced by $L'$.
The full BRST current corresponding to $\gamma$ is given by the formula:
\[
J^\mu(x)\doteq\sum\limits_{\al}\left(\gamma\ph^\al\frac{\partial{L_\Mcal(x)}}{\partial(\nabla_\mu\ph^\al)}+2\nabla_\nu\gamma\ph^\al\frac{\partial{L_\Mcal(x)}}{\partial(\nabla_\mu\nabla_\nu\ph^\al)}-\nabla_\nu\left(\gamma\ph^\al\frac{\partial{L_\Mcal(x)}}{\partial(\nabla_\mu\nabla_\nu\ph^\al)}\right)\right)-K^\mu_{\Mcal}(x)\,,
\]
where $K_{\Mcal}^\mu$ is the divergence term appearing after applying $\gamma$ to $L_\Mcal(f)$. Using this formula we obtain (compare with \cite{Nish,KuOji78,Nak78}):
\be\label{current}
J^\mu=\sqrt{-g}g^{\mu\la}(b_\rho\nabla_\la c^\rho-(\nabla_\la b_\rho) c^\rho)+\al (b^\rho+ic^\al\nabla_\al\bar{c}^\rho)(b_\rho+ic^\al\nabla_\al\bar{c}_\rho)
+i\sqrt{-g}g^{\mu\la}c^\al c^\rho R_{\la\al\rho}^{\phantom{\la\al\rho}\bet}\bar{c}_\bet \,.
\ee
The free BRST current is given by:
\[
J_0^\mu=\sqrt{-g}g^{\mu\la}(b_\rho\nabla_\la c^\rho-(\nabla_\la b_\rho) c^\rho) \,.
\]
For a spacetime $\Mcal$ with compact Cauchy surface $\Sigma$, for any closed 3-form $\beta$ there exists a closed compactly supported 1-form $\eta$ on $\Mcal$ such that $\int_M\eta\wedge\beta=\int_\Sigma\beta$. In this case we can define the BRST charge as:
 \[
 Q\doteq \int_M\eta\wedge J\,
 \]
 and analogously for the free BRST charge $Q_0$.

\end{document}